\documentclass[a4paper,english,cleveref, autoref,numberwithinsect, nolineno]{socg-lipics-v2021}

\usepackage{longtable}

\usepackage{graphicx}
\usepackage{color}
\usepackage{amsmath}
\usepackage{amssymb}
\usepackage{xspace}
\usepackage{epsfig}
\usepackage[dvipsnames]{xcolor}
\usepackage{todonotes}
\usepackage{tikz}

\theoremstyle{plain}

\newtheorem*{supermaintheorem*}{Main Theorem}

\newcommand{\exfoldcov}[3]{{#1}^{#2}_{#3}}
\newcommand{\abs}[1]       {\mm{\left|{#1}\right|}}

\newcommand {\mm}[1] {\ifmmode{#1}\else{\mbox{\(#1\)}}\fi}

\newcommand{\ignore}[1]{}

\newsavebox{\smallProofsym}                            
\savebox{\smallProofsym}                               %
{
\begin{picture}(6,6)                                   %
\put(0,0){\framebox(6,6){}}                            %
\put(0,2){\framebox(4,4){}}                            %
\end{picture}                                          %
}                                                      %

\newcommand{\Rspace} {\mm{{\mathbb R}}}
\newcommand{\Zspace} {\mm{{\mathbb Z}}}

\newcommand{\domain}[3]     {\mm{{\rm dom}_{#1}{({#2},{#3})}}}
\newcommand{\zone}[3]       {\mm{{\rm zone}_{#1}{({#2},{#3})}}}
\newcommand{\ball}[2]       {\mm{{B}{({#2};{#1})}}}

\newcommand{\Radius}[1]   {\mm{\rm Rad}{({#1})}}

\newcommand{\Edist}[2]      {\mm{\|{#1}-{#2}\|}}
\newcommand{\infdist}[2]    {\mm{{\|{#1}-{#2}\|}_{\infty}}}

\newcommand{\dBottleneck}[2]  {\mm{{d_B}{({#1},{#2})}}}
\newcommand{\dFingerprint}[2] {\mm{{d_\infty}{({#1},{#2})}}}
\newcommand{\Fingerprint}[1] {\mm{\Psi}_{#1}}

\newcommand{\probA}[2]      {\mm{{\varphi}_{#1}^{#2}}}
\newcommand{\probB}[2]      {\mm{{\psi}_{#1}^{#2}}}
\newcommand{\Prob}[1]       {\mm{{\rm Prob}{[{#1}]}}}
\newcommand{\volume}[1]     {\mm{{\rm Vol}{[{#1}]}}}

\newcommand{\Isometry}      {\mm{{\rm iso}}}

\newcommand{\card}[1]       {\mm{{\#}{#1}}}

\newcommand{\norm}[1]       {\mm{\|{#1}\|}}
\newcommand{\dd}            {\mm{\delta}}
\newcommand{\ee}            {\mm{\varepsilon}}

\EventEditors{Kevin Buchin and \'{E}ric Colin de Verdi\`{e}re}
\EventNoEds{2}
\EventLongTitle{37th International Symposium on Computational Geometry (SoCG 2021)}
\EventShortTitle{SoCG 2021}
\EventAcronym{SoCG}
\EventYear{2021}
\EventDate{June 7--11, 2021}
\EventLocation{Buffalo, NY, USA}
\EventLogo{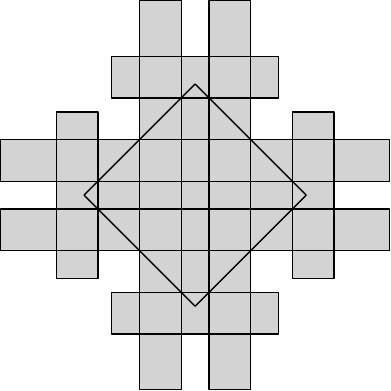}
\SeriesVolume{189}

\ArticleNo{19}

\title{The Density Fingerprint of a Periodic Point Set}

\authorrunning{H. Edelsbrunner, T. Heiss, V. Kurlin, P. Smith, and M. Wintraecken}

\author{Herbert Edelsbrunner}{IST Austria (Institute of Science and Technology Austria),
  Kloster\-neu\-burg, Austria}{Herbert.Edelsbrunner@ist.ac.at}
  {https://orcid.org/0000-0002-9823-6833}{ERC Horizon 2020 `Alpha Shape Theory Extended', no.\ 788183;  FWF `Wittgenstein Prize', no.\ Z 342-N31;  FWF DFG TRR 109 `Discretization in Geometry and Dynamics', no.\ I 02979-N35.}

\author{Teresa Heiss}{IST Austria (Institute of Science and Technology Austria), Kloster\-neu\-burg, Austria}{Teresa.Heiss@ist.ac.at}{https://orcid.org/0000-0002-1780-2689}{ERC Horizon 2020 `Alpha Shape Theory Extended', no.\ 788183.}

\author{Vitaliy Kurlin}{Department of Computer Science, University of Liverpool, Liverpool, United Kingdom}{Vitaliy.Kurlin@liverpool.ac.uk}{https://orcid.org/0000-0001-5328-5351}{EPSRC grant `Application-driven Topological Data Analysis' (EP/R018472/1).}

\author{Philip Smith}{Department of Computer Science, University of Liverpool, Liverpool, United Kingdom}{Philip.Smith@liverpool.ac.uk}{https://orcid.org/0000-0003-3001-0245}{Leverhulme Research Centre for Functional Materials Design at the University of Liverpool, UK.}

\author{Mathijs Wintraecken}{IST Austria (Institute of Science and Technology Austria),
  Kloster\-neu\-burg, Austria}{Mathijs.Wintraecken@ist.ac.at}
  {https://orcid.org/0000-0002-7472-2220}{the European Union's Horizon 2020 research and innovation programme under the Marie Sk{\l}odowska-Curie grant agreement No. 754411. }

\Copyright{Herbert Edelsbrunner, Teresa Heiss, Vitaliy Kurlin, Philip Smith, and Mathijs Wintraecken}



\funding{}

\acknowledgements{The authors thank Janos Pach for insightful discussions on the topic of this paper, Morteza Saghafian for finding the one-dimensional counterexample mentioned in Section~\ref{sec:Completeness}, and Larry Andrews for generously sharing his crystallographic perspective.} 

\begin{document}
\maketitle

\begin{abstract}
  Modeling a crystal as a periodic point set, we present a fingerprint consisting of density functions that facilitates the efficient search for new materials and material properties.
  We prove invariance under isometries, continuity, and completeness in the generic case, which are necessary features for the reliable comparison of crystals.
  The proof of continuity 
  integrates methods from discrete geometry and lattice theory, while the proof of generic completeness combines techniques from geometry with analysis.
  The fingerprint has a fast algorithm based on Brillouin zones and related inclusion-exclusion formulae.
  We have implemented the algorithm and describe its application to crystal structure prediction.
\end{abstract}

\section{Introduction}
\label{sec:Introduction}

This paper considers a deceptively simple question:  \emph{given periodic point sets (crystals) in $\Rspace^3$, determine how close the sets are to being isometric.}
In other words, how much do the points need to be perturbed to allow for a rigid transformation between the two sets?
More generally, we may ask for the organization of a collection of periodic sets that facilitates efficient search.
In this context, a \emph{periodic (point) set} is the Minkowski sum of a \emph{lattice} and a \emph{motif}.
The lattice is spanned by three linearly independent vectors, and the motif is a finite set of points in the \emph{unit cell}, which is the parallelepiped whose edges are translates of the three vectors. 
The main reason for the difficulty of the question is the complicated nature of the continuous space of isometry classes of periodic sets:
\medskip \begin{itemize} 
  \item There is no method for choosing a unique basis for a lattice in a continuous manner.  Indeed, continuity contradicts uniqueness as we can continuously deform a basis to a different basis of the same lattice.
  For example, the Niggli reduced cell \cite{Nig28} is unique but not continuous with respect to perturbations of the lattice \cite{ABP80}.
  \item Crystallographers often use the symmetry group of crystals, which define a stratification of the space of isometry classes.
  Belonging to a given stratum is however not a continuous property.
  \item Small perturbations of a periodic set can significantly change the lattice with respect to which it is periodic.
\end{itemize} \medskip
Periodic sets are usually given by a basis of a lattice and the motif within the unit cell spanned by this basis.
As explained above, even for extremely similar periodic sets, their lattices, their bases, and thus their motifs can look completely different,
making a direct comparison between the motifs impossible.
It is therefore advisable not to compute the distances between the isometry classes but instead map the sets to a less complicated metric space.
This is the approach we take in this paper.
Part of the challenge is to determine which properties this map should possess, and how to balance its mathematical properties with efficient computability.
As in many applications, we consider false positives in the comparison of two crystals less problematic than false negatives: a large distance in the metric space should imply that the two periodic sets are far from isometric, while a small distance should indicate a high chance that the sets are indeed close to being isometric.

\medskip
The main contribution of this paper is a candidate solution, which we refer to as the \emph{density fingerprint map}.
For different non-negative integers, $k$, and for different non-negative radii, $t$, it maps the periodic set to the probability that a random point in the unit cell is at distance at most $t$ from exactly $k$ points in the periodic set; see Definition \ref{def:density_functions_and_fingerprint}.
\begin{supermaintheorem*}
\label{thm:main_theorem_introduction}
  The density fingerprint maps a periodic set in $\Rspace^3$ to a series of density functions that satisfy the following properties:
  \medskip \begin{enumerate}
    \item the map is invariant under isometries (rigid motions and reflections) of space;
    \item the map is Lipschitz continuous with respect to small perturbations of the points, with the Lipschitz constant depending on the packing and covering radii of the periodic set;
    \item the map is generically complete: for a dense and open subset of the space of periodic sets, the isometry class of the periodic set is uniquely determined.
  \end{enumerate}
\end{supermaintheorem*}
\medskip
Short of proving completeness beyond the generic case, we leave the completeness of the density fingerprint map for all periodic sets as an open question.
Indeed, the authors of this paper failed to produce a counterexample to general completeness, but not because of a lack of trying.
The Main Theorem generalizes to arbitrary finite dimensions.

\medskip
The quest for a fingerprint map is motivated by the study of crystals, for which the configuration of atoms is important for their chemical properties.
Quantifying the similarity between crystals has the potential to greatly improve the practice of \emph{Crystal Structure Prediction}, which has come to rely on high-performance computing for simulating millions of structures.
Many of them are similar to each other, and very few are eventually produced in the laboratory.
An effective notion of similarity would allow an improved organization of the structures and lead to vast savings of supercomputing time currently wasted on redundant simulations.
The prior work in this area is best summarized by listing the software systems currently used in practice:
\textsc{Compack} \cite{ChMo05}, \textsc{Mercury} \cite{MEMP06}, and \textsc{Compstru} \cite{FOTPA16}.
These systems are of great help in comparing crystals, but they employ heuristics like cut-offs and tolerances, which come with the usual drawbacks.
It is our ambition to develop the mathematical and computational foundations needed to overcome the current deficiencies.

\medskip
\textbf{Outline.}
Section \ref{sec:Background} provides the necessary notation and terminology for lattices, periodic sets, and isometries.
Section \ref{sec:Fingerprint} introduces the density functions for a periodic set and the corresponding density fingerprint map.
Section \ref{sec:Stability} proves that the density fingerprint map is continuous with respect to perturbations of the periodic set.
Section \ref{sec:Completeness} proves that the density fingerprint map is complete for generic periodic sets.
Section \ref{sec:Computation} explains how the density fingerprint is computed using the Brillouin zones of the points.
Section \ref{sec:Appication} describes a preliminary application of the density fingerprint map to Crystal Structure Prediction.
Section \ref{sec:Discussion} concludes the paper.

\section{Background} 
\label{sec:Background}

We cover two topics: locally finite point sets modeling crystals and transformations between them.

\subsection{Delone Sets and Periodic Sets}
\label{sec:Delone_Periodic}

We recall that $A \subseteq \Rspace^3$ is \emph{locally finite} if any compact
subset of $\Rspace^3$ contains only finitely many points of $A$.
It is a \emph{Delone set} \cite{DLS98} if
there exist $r, R > 0$ such that every open ball of radius $r$ contains
at most one point of $A$ and every closed ball of radius $R$ contains
at least one point of $A$.
In other words, no two points of $A$ can be closer than $2r$
and no point of $\Rspace^3$ can be further from $A$ than $R$.
We refer to the largest such $r$ as the \emph{packing radius}
and the smallest such $R$ as the \emph{covering radius} of $A$.
A Delone set is necessarily infinitely large and its points
are in a sense evenly spread out over the entire Euclidean space.

\medskip
We get an important subclass starting with three linearly independent vectors, $v_1, v_2, v_3 \in \Rspace^3$, which we call a \emph{basis}.
The set of integer combinations is the \emph{lattice}, $\Lambda$, and the set of real combinations with coefficients in $[0, 1)$ is the \emph{unit cell}, $U$,
the vectors span:
\begin{align}
  \Lambda  &=  \{ n_1 v_1 + n_2 v_2 + n_3 v_3 \mid n_1,n_2,n_3 \in \Zspace \} , \\
  U        &=  \{ r_1 v_1 + r_2 v_2 + r_3 v_3 \mid 0 \leq r_1, r_2, r_3 < 1 \} .
\end{align}
We call any finite set $M \subseteq U$ a \emph{motif} and $M + \Lambda = \{x+v \mid x \in M, v \in \Lambda \}$ a \emph{periodic (point) set}.
By construction, $M + \Lambda + v = M + \Lambda$ for every $v \in \Lambda$, and if this exhausts all translations that keep $M + \Lambda$ invariant, then $U$ is a \emph{primitive unit cell} of $M + \Lambda$.
Since the International Union of Crystallography (IUCr) allows lattices that are spanned by fewer than three linearly independent vectors, it calls lattices as defined above \emph{full}.
We observe that $M + \Lambda$ is a Delone set if $\Lambda$ is full and not if $\Lambda$ is not full, and therefore we will be exclusively interested in full lattices and so will assume as much without mentioning the term.

\medskip
It is important to keep in mind that the basis and therefore the primitive unit cell are not unique.
This is illustrated in Figure \ref{fig:unit-cell},
which shows three of the infinitely many bases of the hexagonal lattice:  $a$ together with $b-a$, $b$, or $b+a$.
Applying Niggli's algorithm for the \emph{Niggli reduced cell} \cite{Nig28} to this particular lattice, there is an ambiguity between the bases $\{a, b\}$ and $\{a, b-a\}$, because the projections of $b$ and $b-a$ onto the line of $a$ both have length $\tfrac{1}{2} \norm{a}$.
The tie can be broken by preferring $b$, but this causes a discontinuity in the construction of the reduced unit cell.

\begin{figure}[hbt]
  \centering \resizebox{!}{1.8in}{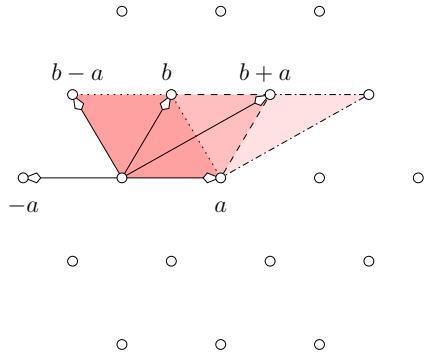}
  \caption{The bases $\{a, b-a\}$, $\{a, b\}$, $\{a, b+a\}$ of the hexagonal lattice, and the corresponding unit cells drawn as shaded parallelograms.}
  \label{fig:unit-cell}
\end{figure}

\subsection{Rigid Motions and Isometries}
\label{sec:RigidMotion_Isometry}

Rather than a fixed set in $\Rspace^3$, we often consider the class of sets that are equivalent under a particular type of transformation.
For example, a \emph{rigid motion} is a map $\Rspace^3 \to \Rspace^3$ that is composed of a rotation and a translation.
It preserves distances between pairs of points as well as orientations of ordered triplets of points.
An \emph{isometry} is a rigid motion possibly composed with a reflection, and so preserves distances but not necessarily orientations.
This is the most relevant group of transformations to this paper as we model crystals by isometry classes of periodic sets.

\section{The Density Fingerprint and its Invariance}
\label{sec:Fingerprint}

A continuous invariant for comparing crystals is the density, defined as the total volume of balls centered at points in the motif divided by the volume of the unit cell.
To avoid the choice of radii, we grow the balls continuously and simultaneously from their centers and get a $1$-dimensional function rather than a single number.
There are still many periodic sets this function cannot distinguish, for example any hexagonal close packing from the face-centered cubic lattice.
We therefore add information by distinguishing points covered by a different number of balls.
\begin{definition}[Density Functions and Fingerprint]
  \label{def:density_functions_and_fingerprint}
  Let $A = M + \Lambda \subset \Rspace^3$ be a periodic set and write $A (t)$ for the collection of closed balls, $\ball{t}{a}$, of radius $t \geq 0$ centered at the points $a \in A$.
  The \emph{$k$-fold cover} of $A(t)$, denoted $\bigcup^k A(t)$, consists of all points $x \in \Rspace^3$ contained in $k$ or more of these balls.
  The fractional volume of the $k$-fold cover, $ \probA{k}{A}(t) = \volume{U \cap \bigcup^k A(t)} / \volume{U}$, is also the probability that a point chosen uniformly at random within a unit cell, $U$, belongs to at least $k$ balls, and subtracting the fractional volume of the $(k+1)$-fold cover, we get the probability that the random point belongs to exactly $k$ balls:
  \begin{align}
    \probA{k}{A}(t)  &=  \Prob{x \in \ball{t}{a} \mbox{\rm ~for~} k \mbox{\rm ~or more points~} a \in A} ; 
    \label{eqn:probA} \\
    \probB{k}{A}(t) &= \probA{k}{A}(t)-\probA{k+1}{A}(t)  = \Prob{x \in \ball{t}{a} \mbox{\rm ~for exactly~} k    \mbox{\rm ~points~} a \in A} .
    \label{eqn:probB}
  \end{align}
  We call $\probB{k}{A} \colon [0, \infty) \to [0,1]$ the \emph{$k$-th density function} of $A$.
  The \emph{density fingerprint} of $A$ is the vector of density functions:  $\Fingerprint{}(A) = ( \probB{0}{A}, \probB{1}{A}, \ldots, \probB{k}{A}, \ldots )$, and
  $A \mapsto \Fingerprint{}(A)$ is the \emph{density fingerprint map}.
\end{definition}
See Figure \ref{fig:hexagonalsquare}, which illustrates the density functions for the hexagonal and the square lattices in $\Rspace^2$.
Note that the density fingerprint is an isometry invariant and that it neither depends on the lattice used to write $A$ as a periodic set, nor on its basis.
\begin{lemma}[Invariance under Isometries]
  Let $A \subseteq \Rspace^3$ be a periodic set, and let $Q \subseteq \Rspace^3$ be isometric to $A$.
  Then $\Fingerprint{}(A) = \Fingerprint{}(Q)$.
\end{lemma}
\begin{proof}
  Let now $\Isometry \colon \Rspace^3 \to \Rspace^3$ be the isometry for which $Q = \Isometry (A)$, and note that it also maps $A(t)$ to $Q(t)$ and $\bigcup^k A(t)$ to $\bigcup^k Q(t)$ for every $k \geq 0$.
  It follows that $\probB{k}{A}(t) = \probB{k}{Q}(t)$, for every $k \geq 0$, and therefore $\Fingerprint{}(A) = \Fingerprint{}(Q)$, as claimed.
\end{proof}

While the fingerprint map is not invariant under similarities, we can write
$\Fingerprint{}(sA)  =  ( \probB{0}{A} \circ s, \probB{1}{A} \circ s, \ldots, \probB{k}{A} \circ s, \ldots )$,
in which $s(t) = st$ scales the radius.
It would therefore be easy to construct a fingerprint map that is invariant under similarities, 
namely by normalizing the radius,
e.g.\ by letting the radius be $tr$, in which $r$ is the packing radius of $A$.

\begin{figure}[ht]
  \parbox{59mm}{
  \includegraphics[height=28mm]{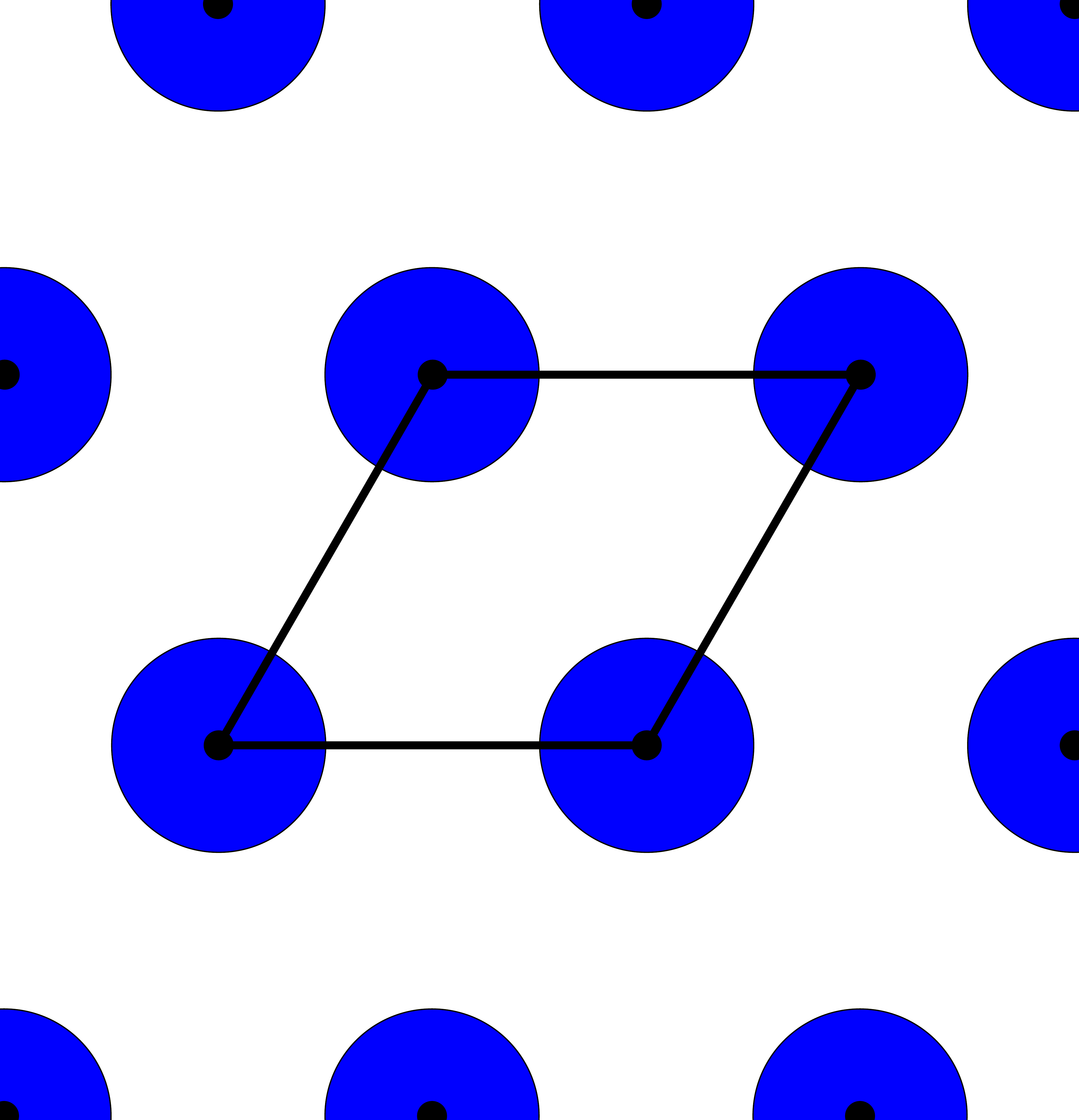}
  \hspace*{1mm}
  \includegraphics[height=28mm]{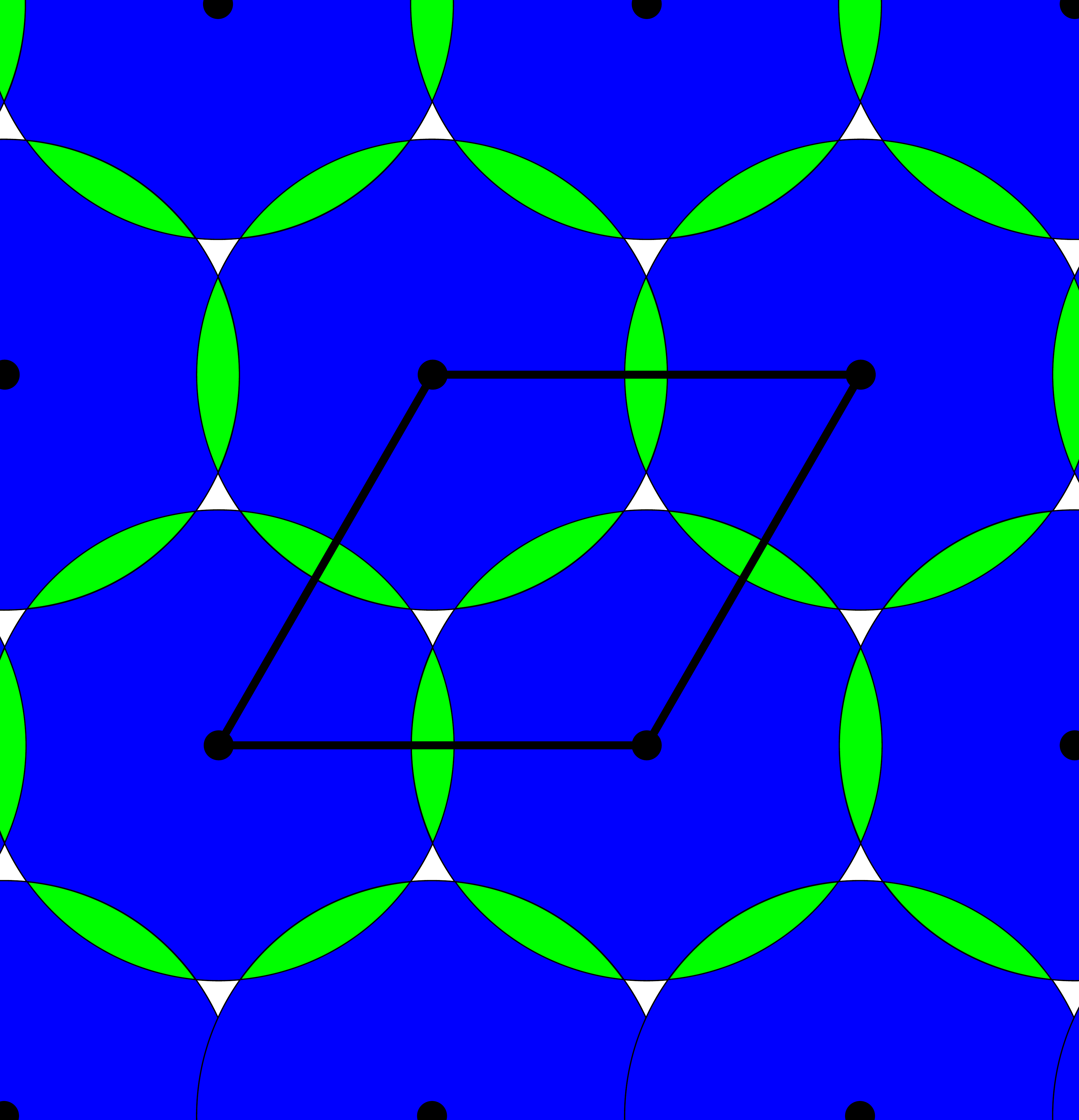}
  \medskip
  \includegraphics[height=28mm]{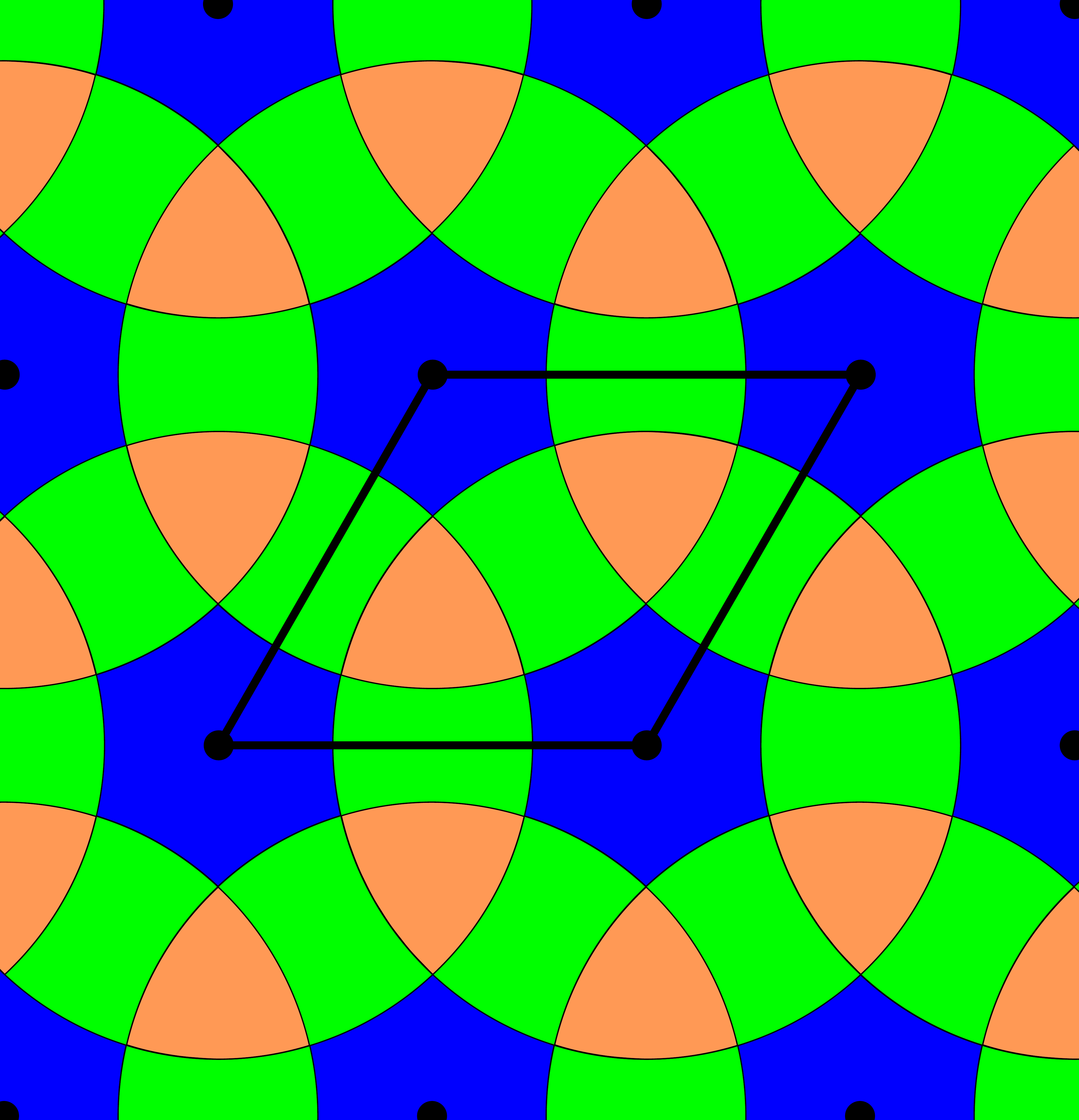}
  \hspace*{3mm}
  \includegraphics[height=28mm]{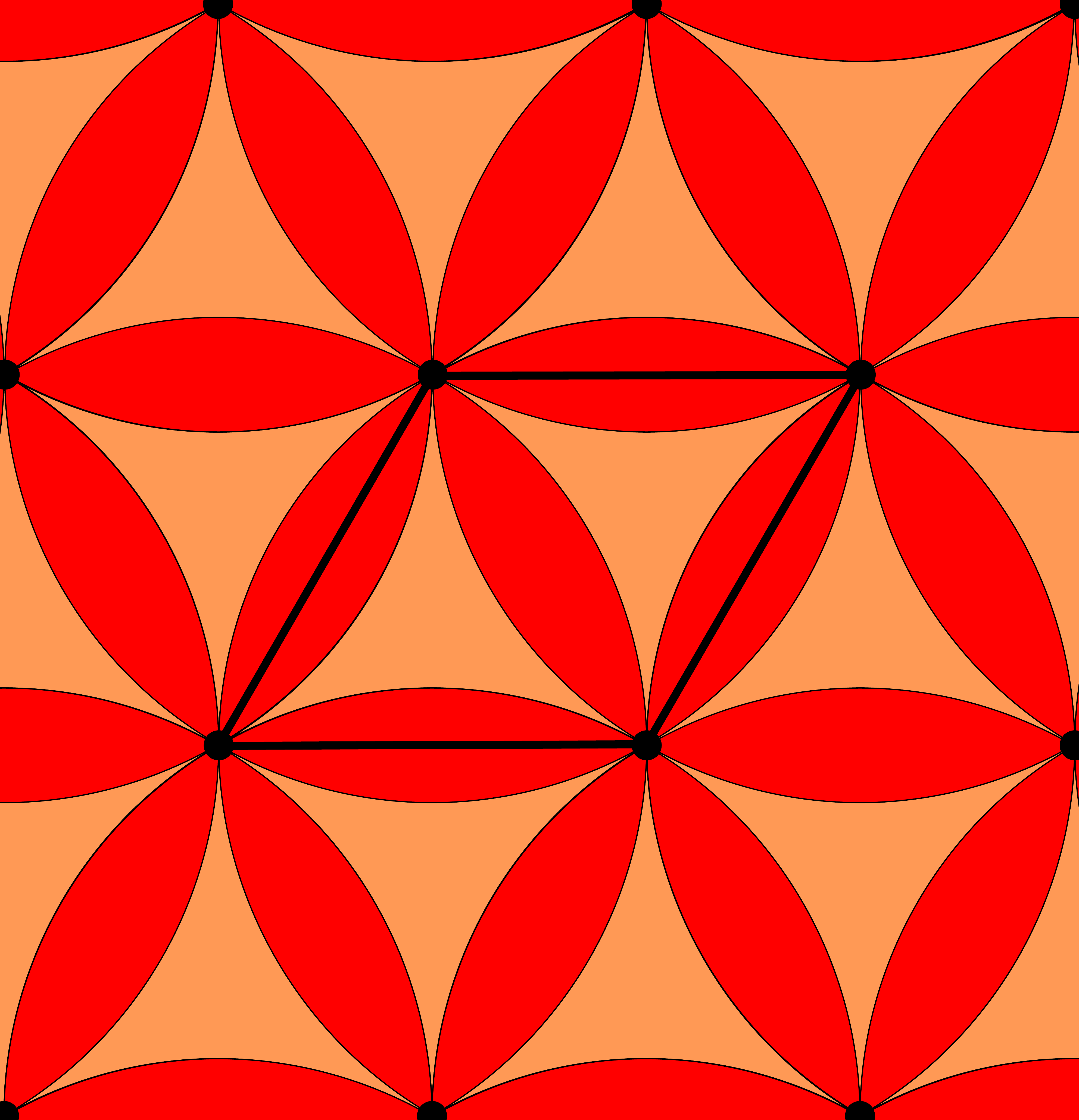}}
  \parbox{80mm}{
  \includegraphics[width=80mm]{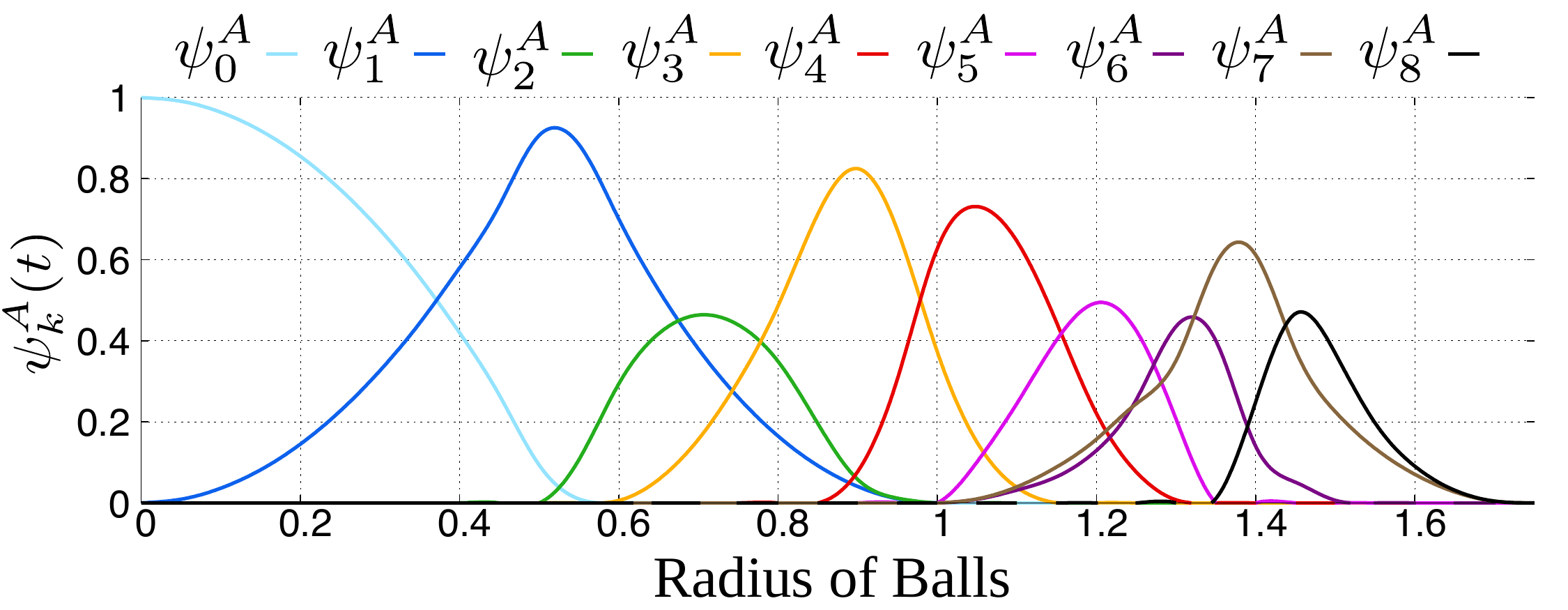}
  \includegraphics[width=80mm]{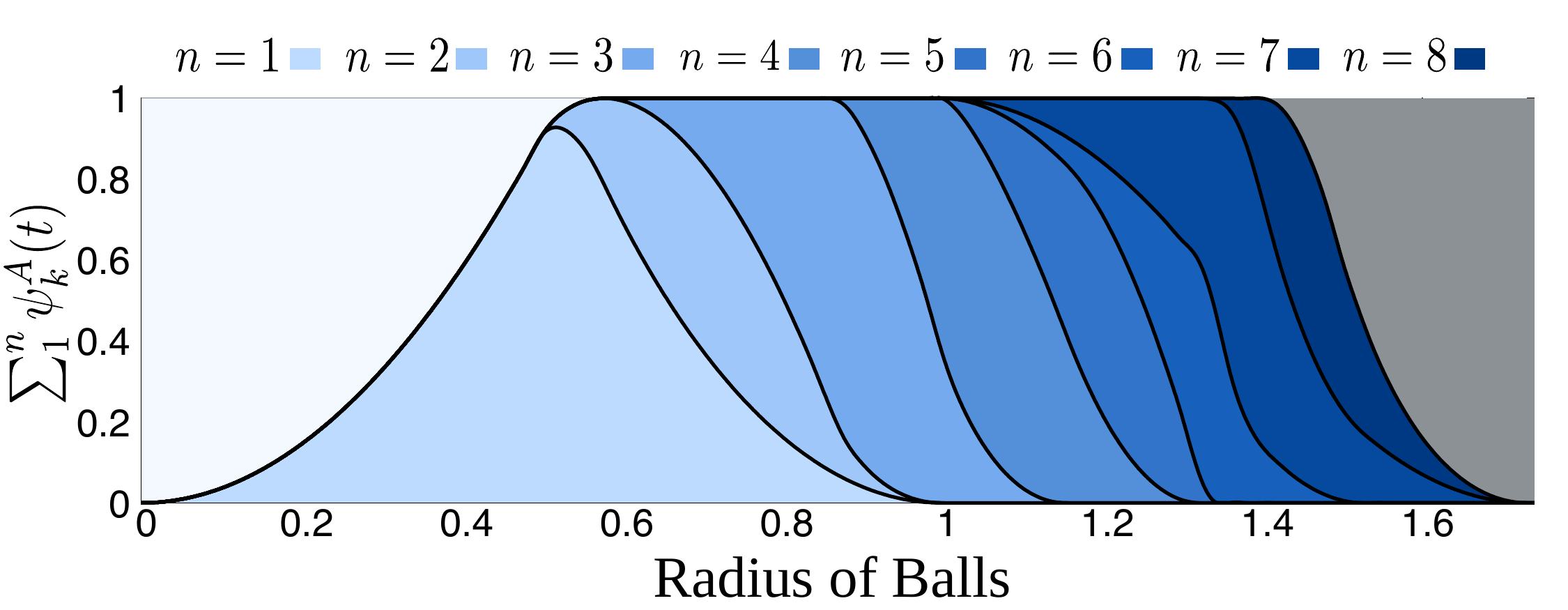}}
  
  \vspace*{1mm}
  
  \parbox{59mm}{
  \includegraphics[height=28mm]{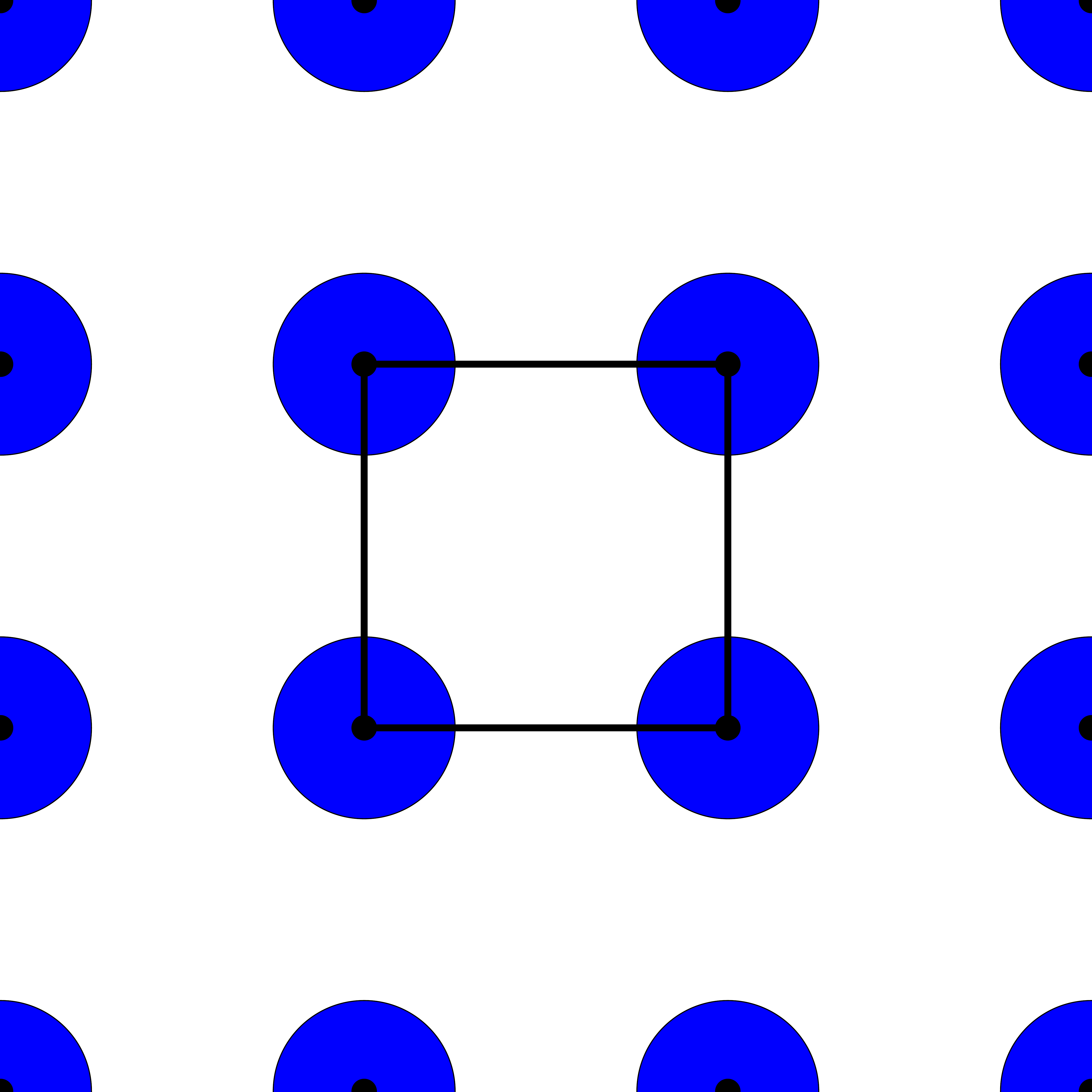}
  \hspace*{1mm}
  \includegraphics[height=28mm]{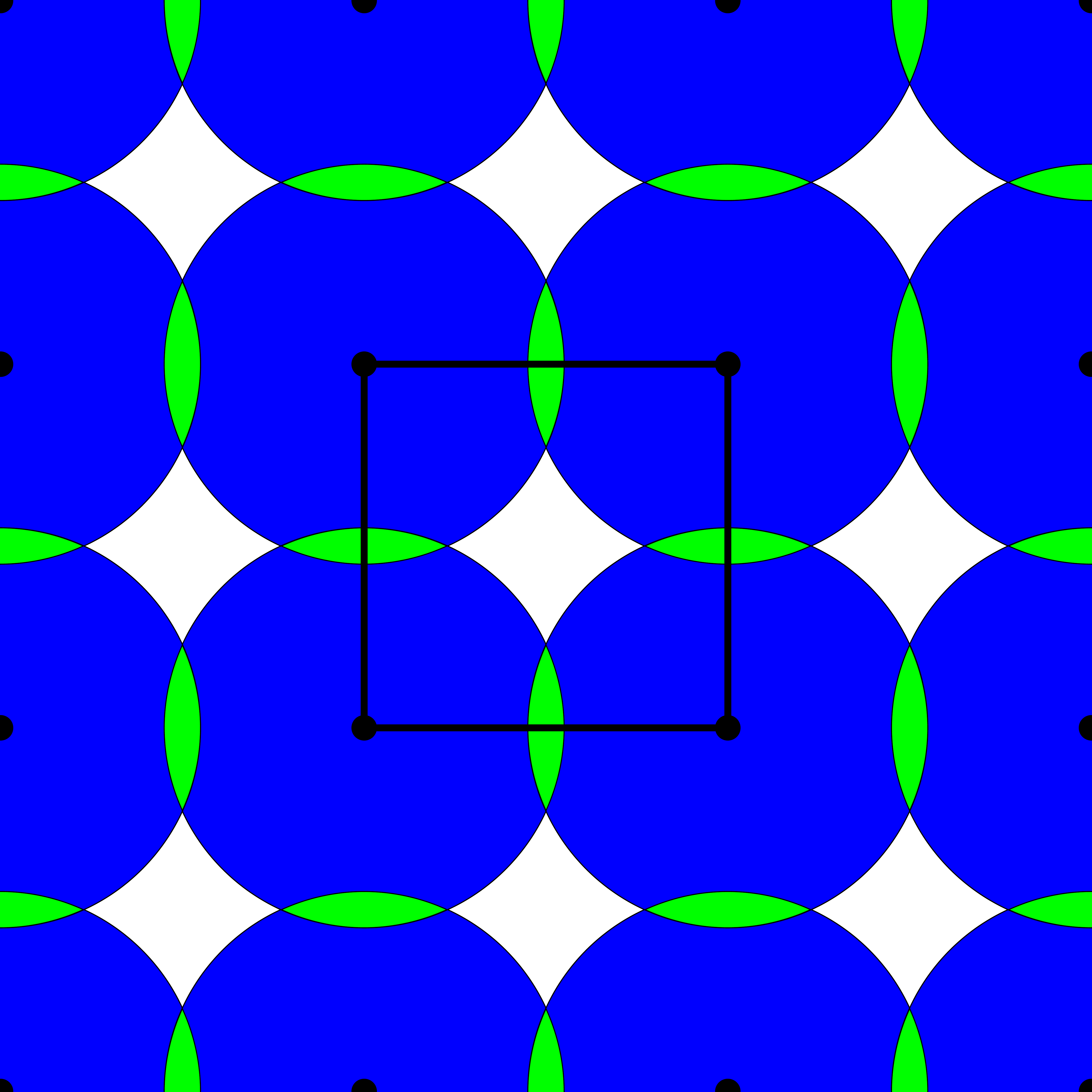}
  \medskip
  \includegraphics[height=28mm]{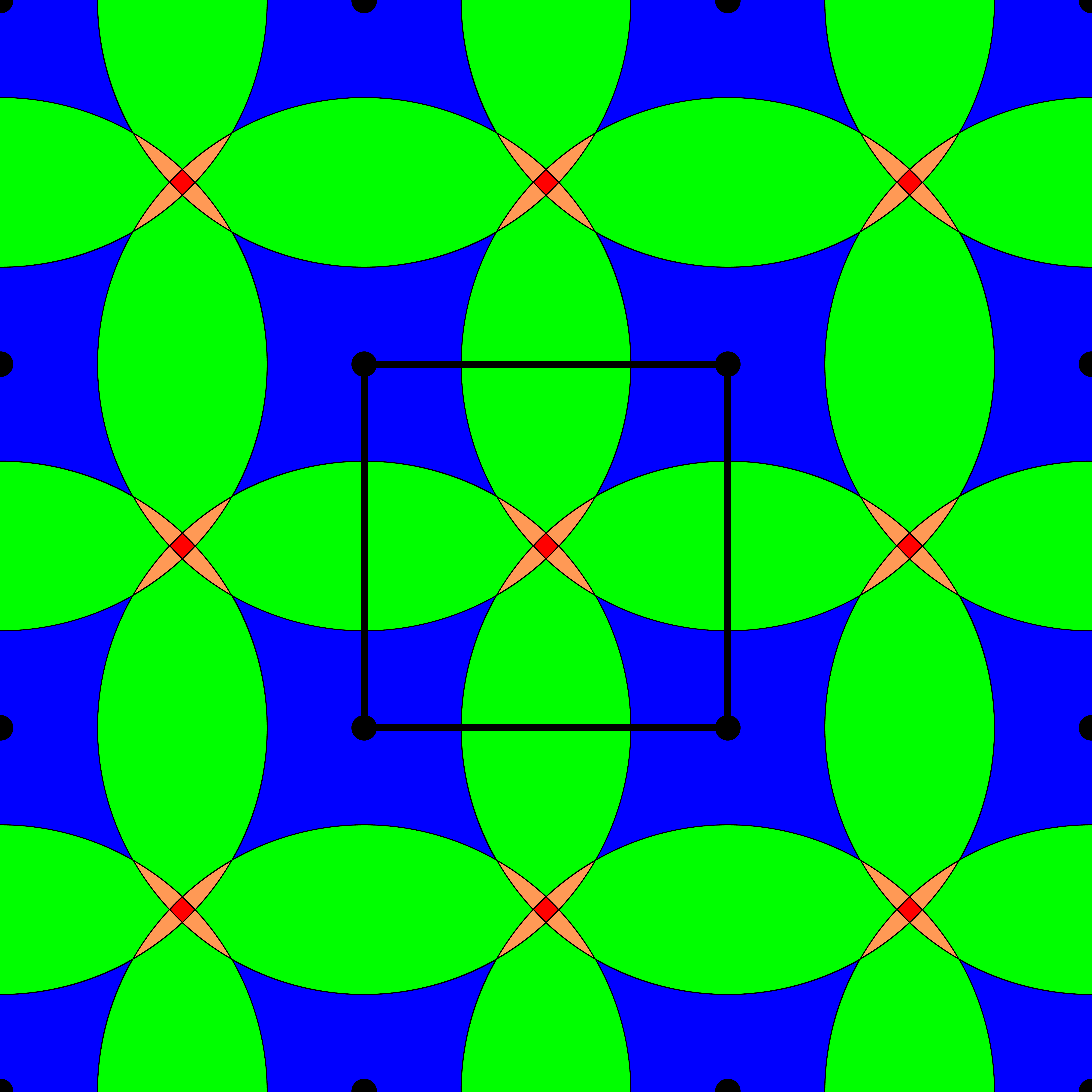}
  \hspace*{1mm}
  \includegraphics[height=28mm]{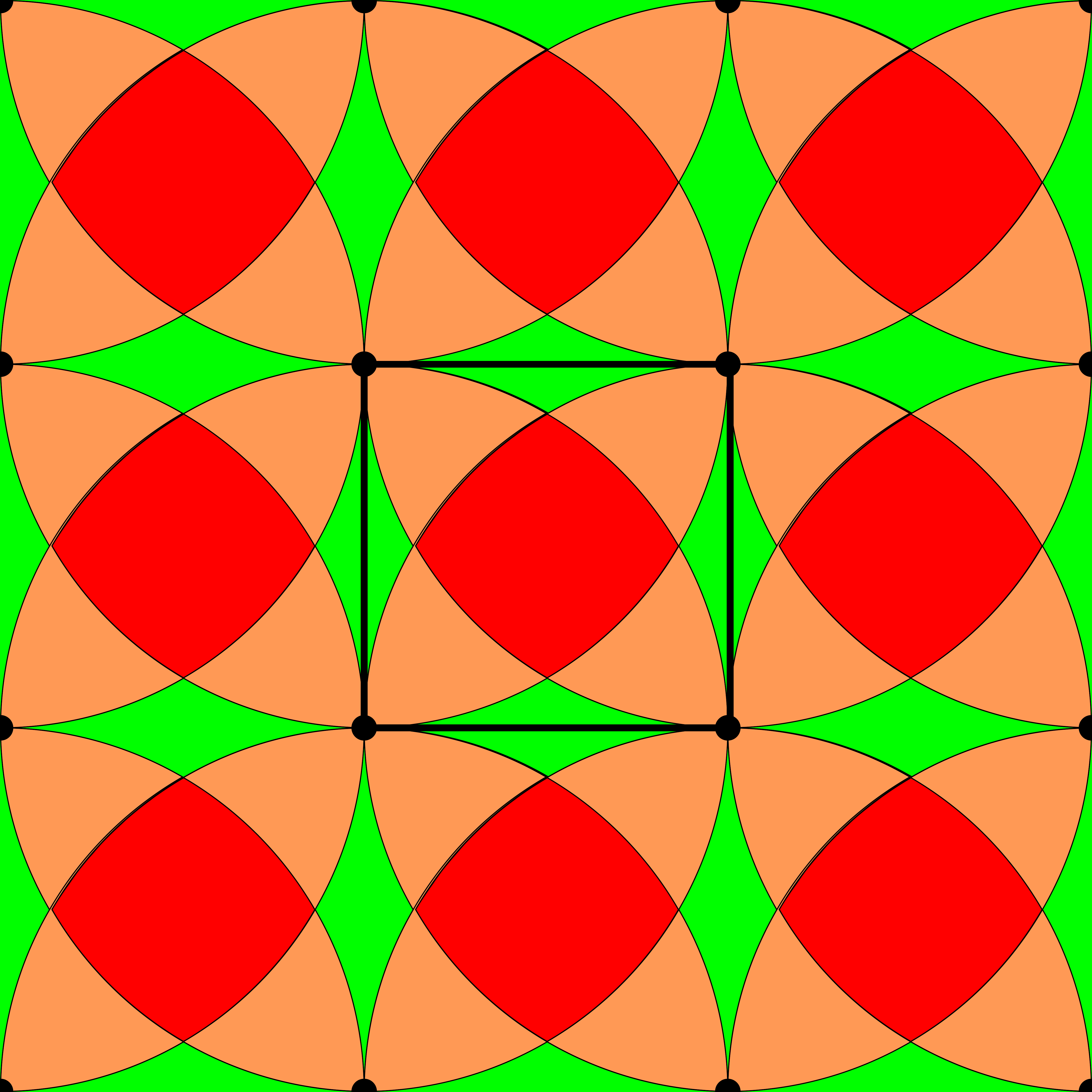}}
  \parbox{80mm}{
  \includegraphics[width=80mm]{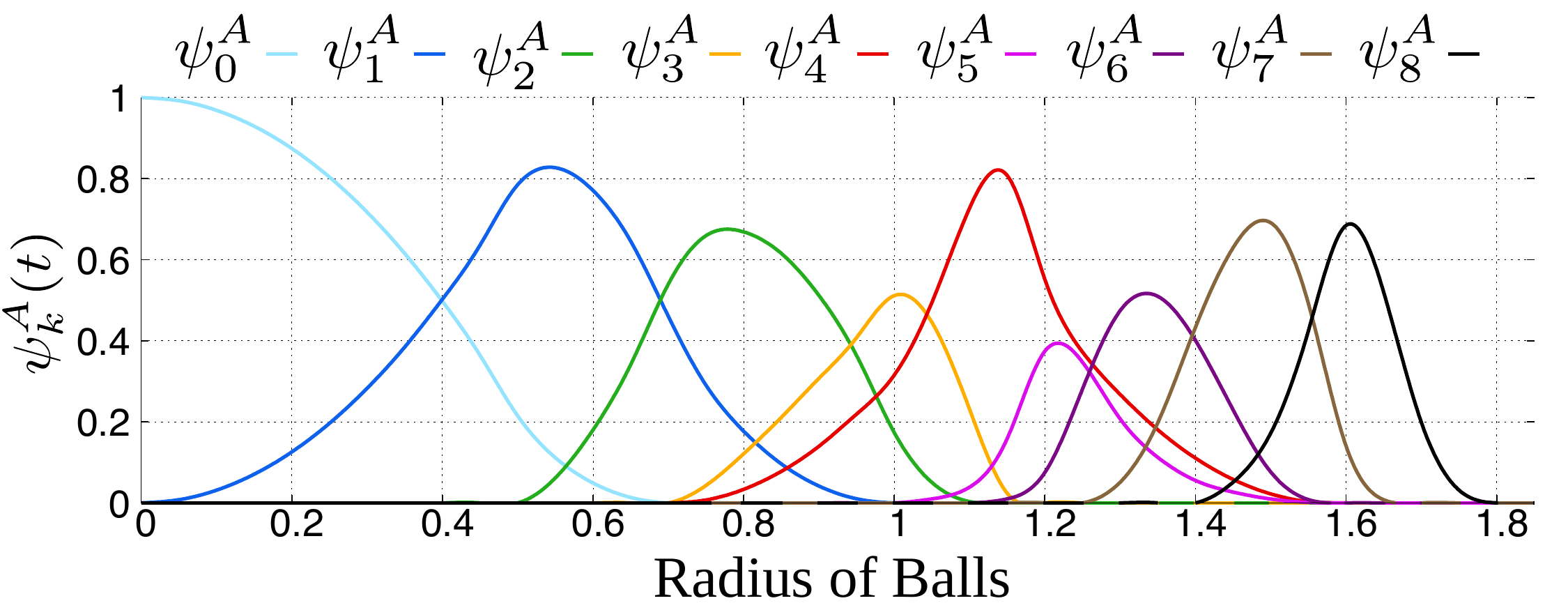}
  \includegraphics[width=80mm]{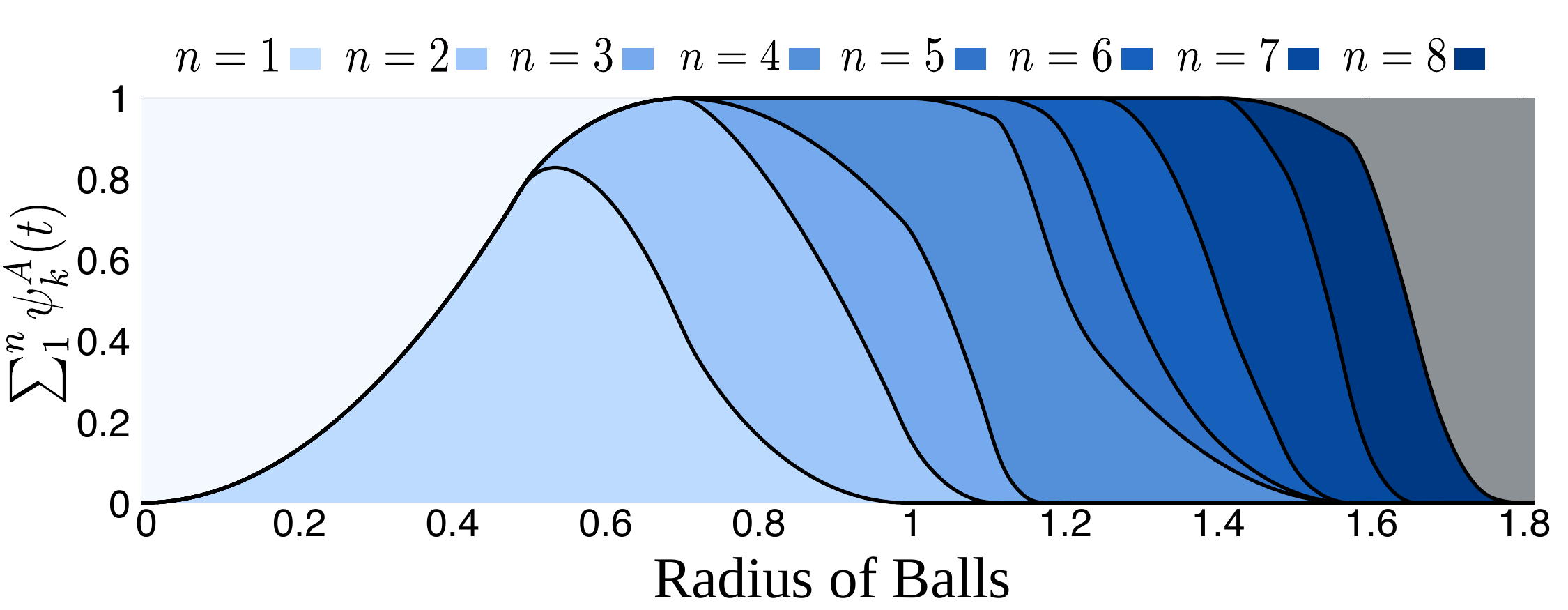}}
  \caption{The density fingerprint map of the hexagonal lattice on the \emph{top} and, for comparison, of the square lattice on the \emph{bottom}.
  \emph{Left}: the $k$-fold covers of the two sets for four different radii each: $t=0.25, 0.55, 0.75, 1.00$.
  \emph{Right}: the graphs of the respective first nine density functions above the corresponding \emph{densigram}, in which the zeroth function can be seen upside-down and the remaining density functions are accumulated from left to right.}
  \label{fig:hexagonalsquare}
\end{figure}

\section{Continuity}
\label{sec:Stability}

We prove that the density fingerprint map is Lipschitz continuous with respect to small perturbations of the points. 
To formalize this result, we introduce distances between periodic sets and between density fingerprints.
For sets $A, Q \subseteq \Rspace^3$ of equal cardinality, the \emph{(Euclidean) bottleneck distance} is the infimum, over all bijections, $\gamma \colon A \to Q$, of the supremum Euclidean distance between matched points, and for density fingerprints, $\Fingerprint{} (A), \Fingerprint{} (Q)$, we use the supremum of the weighted infinity norms of the differences between corresponding density functions:
\begin{align}
  \dBottleneck{A}{Q}  &=  \inf_{\gamma \colon A \to Q} \sup_{a \in A} \Edist{a}{\gamma(a)}_2 , \\
  \dFingerprint{\Fingerprint{}(A)}{\Fingerprint{}(Q)}  &=  \sup_{k \geq 0} \tfrac{1}{\sqrt[3]{k+1}^2}
    \infdist{\probB{k}{A}}{\probB{k}{Q}} .
\end{align}
Note the damping of the difference between corresponding density functions.
The reason for it is technical and related to the fact that density functions with higher $k$ tend to vanish at later values of $t$. 
As a consequence, the sensitivity of the density function to any perturbation increases with growing $k$, and the damping compensates for this tendency.
Before proving Lipschitz continuity,
we show that two periodic sets with small bottleneck distance between them necessarily have a common lattice.
\begin{lemma}[Common Lattice]
  \label{lem:common_lattice}
  Let $A, Q$ be periodic sets in $\Rspace^3$, and let $r_Q > 0$ be the packing radius of $Q$.
  If $\dBottleneck{A}{Q} < r_Q$, then there is a lattice $\Lambda$ with unit cell $U$ in $\Rspace^3$ such that $\card{(A \cap U)} = \card{(Q \cap U)}$ and $A = (A \cap U) + \Lambda$ and $Q = (Q \cap U) + \Lambda$.
\end{lemma}
\begin{proof}
  Since $A, Q \subseteq \Rspace^3$ are periodic, there are lattices with unit cells such that $A = (A \cap U_A) + \Lambda_A$ and $Q = (Q \cap U_Q) + \Lambda_Q$.
  To get a contradiction, we assume that there is however no common lattice for $A$ and $Q$.
  Equivalently, $\Lambda_A \cap \Lambda_Q$ is a lattice of dimension at most $2$.
  Therefore there exists a basis vector, $v$, of $\Lambda_A$ such that $n v \in \Lambda_Q$ implies $n = 0$.
  Picking a point $a \in A$, we consider the infinitely many points $a(n) = a + nv$, with $n \in \Zspace$.
  For each $a(n)$, let $q(n)$ be the point in $\Lambda_Q$ such that $a(n) \in q(n) + U_Q$, and define $b(n) = a(n) - q(n)$, which we note belongs to $U_Q$.
  
  \smallskip
  There are infinitely many pairwise different points $b(n)$ in the unit cell, and it suffices to prove that at least one is at distance larger than $\dd = \dBottleneck{A}{Q}$ from all points in $Q$.
  To see this, let $b(i)$ and $b(j)$ be at distance less than $\ee = r_Q - \dd$ from each other, and note that $b(i+n[j-i]) = b(i) + n[b(j)-b(i)]$, for $n \in \Zspace$, provided the point on the right-hand side of the equation belongs to $U_Q$.
  In other words, we have an entire line of points with distance less than $\ee$ between contiguous points.
  The gap between balls of radius $\dd$ centered at the points in $Q$ is at least $2 \ee$, which implies that at least one of the points on the line is outside all such balls.
  This contradicts the assumption that the bottleneck distance between $A$ and $Q$ is $\dd = r_Q - \ee$.
  The existence of a common lattice of $A$ and $Q$ follows.
\end{proof}

The proof of Lipschitz continuity makes use of the common lattice of the sets before and after the perturbation.
We therefore formulate the claim assuming that the bottleneck distance between the two sets is less than the packing radii.
\begin{theorem}[Fingerprint Continuity]
  \label{thm:fingerprint_stability}
  Let $A,Q$ be periodic sets in $\Rspace^3$,
  both with packing radius at least $r > 0$ and with covering radius at most $R < \infty$.
  If $\dd = \dBottleneck{A}{Q}  < r$, then
  there exists a constant $C = C(r, R)$ such that
  $\dFingerprint{\Fingerprint{}(A)}{\Fingerprint{}(Q)}  \leq  C \cdot \dBottleneck{A}{Q}$.
\end{theorem}
\begin{proof}
  By Lemma \ref{lem:common_lattice}, there is a lattice, $\Lambda \subseteq \Rspace^3$, that is common to both sets, $A$ and $Q$, and we write $U$ for the corresponding unit cell.
  Let $\gamma \colon A \to Q$ be a bijection such that $\dBottleneck{A}{Q}$ is the supremum Euclidean distance between corresponding points, let $k$ be a non-negative integer, and let $t$ be a positive real number.
  We need an upper bound for
  \begin{align}
  	\abs{\probB{k}{A}(t) - \probB{k}{Q}(t)} &= \frac{   \abs{ \volume{\exfoldcov{A}{k}{t} \cap U} - \volume{\exfoldcov{Q}{k}{t} \cap U} }   }{\volume{U}},
  	\label{eqn:to_bound}
  \end{align}
  in which $\exfoldcov{A}{k}{t} = \bigcup^k A(t) \setminus \bigcup^{k+1} A(t)$ consists of all points $x \in \Rspace^3$ contained in exactly $k$ balls of $A(t)$, and similarly for $\exfoldcov{Q}{k}{t}$.
  As a first step, we find an upper bound on the numerator, $\Delta$, for the case in which $\gamma$ is the identity except for one point, $a \in M$, which it maps to $q = \gamma(a) \in \ball{\dd}{a}$; that is: $Q=A \setminus (a+\Lambda) \cup (q+\Lambda)$.
  A point $x \in \Rspace^3$ is possibly covered by a different number of balls before and after the perturbation only if $x \in [\ball{t}{a} \ominus \ball{t}{q}] +\Lambda$, with $\ominus$ denoting the symmetric difference.
  Observe that this set is contained in $[ \ball{t+\tfrac{\dd}{2}}{\tfrac{a+q}{2}} \setminus \ball{t-\tfrac{\dd}{2}}{\tfrac{a+q}{2}} ] +\Lambda$.
  Hence,
  \begin{align}
    \Delta &\leq \volume{ \ball{t}{a} \ominus \ball{t}{q}} \leq \frac{4 \pi}{3} \left[ \left( t+\tfrac{\dd}{2} \right)^3 - \left( t-\tfrac{\dd}{2} \right)^3\right]= \frac{4 \pi}{3} \left[ 3 \dd t^2 + \tfrac{1}{4} \dd^3 \right].
    \label{eqn:error_per_perturbation}
  \end{align}
  Perturbing one point of $M$ after the other, we can bound the error by \eqref{eqn:error_per_perturbation} each time.
  Using the intensity,
  $\rho = \card{M} / \volume{U}$, this implies
  \begin{align}
    \label{eqn:bound_for_k_t}
    \abs{\probB{k}{A}(t) - \probB{k}{Q}(t)} &\leq \rho \Delta \leq \rho \frac{4 \pi}{3} \left[ 3 \dd t^2 + \tfrac{1}{4} \dd^3 \right].
  \end{align}
  We can eliminate the dependence on $t$ by observing that for each $k$ there is a value of $t$ beyond which the $k$-th density functions of $A$ and $Q$ vanish.
  To determine this value, consider a point $y \in \Rspace^3$ and the sets $A \cap \ball{t}{y}$ and $Q \cap \ball{t}{y}$.
  By the definition of $R$, the balls of radius $R$ centered at the points of $A$ cover $\ball{t-R}{y}$, and similarly for $Q$.
  It follows that the two sets contain at least $(t/R - 1)^3$ points each. Setting $k+1 \leq (t/R - 1)^3$, we see that for $t \geq R \sqrt[3]{k+1} + R$, both sets have at least $k+1$ points each.
  Equivalently, $y$ is covered by at least $k+1$ balls of radius $t$.
  Since this holds for every point $y \in \Rspace^3$, we have $\probB{k}{A}(t) = \probB{k}{Q}(t) = 0$ for all $t \geq R \sqrt[3]{k+1} + R$.
  Note that $R \sqrt[3]{k+1} + R \leq 2 R \sqrt[3]{k+1}$ for all $k \geq 0$. Replacing $t$ in \eqref{eqn:bound_for_k_t} by the latter bound, we get
  \begin{align}
    \frac{1}{\sqrt[3]{k+1}^2} \norm{\probB{k}{A} - \probB{k}{Q}}_\infty \leq 16 \pi \rho R^2 \dd  + \frac{\pi}{3} \rho \dd^3
    \leq  \frac{12 R^2}{r^3} \dd + \frac{1}{4 r^3} \dd^3 ,
    \label{eqn:explicit_lipschitz}
  \end{align}
  in which we use $\rho \frac{4 \pi}{3} r^3 \leq 1$ to get the final inequality.
  Using $\dd^2 < r^2 < R^2$, this gives $C = 13 R^2 / r^3$ as an upper bound for the Lipschitz constant.
\end{proof}

Figure \ref{fig:cloud2D_packing2} illustrates Theorem \ref{thm:fingerprint_stability} for a periodic set, $A$, and its perturbation, $Q$, in $\Rspace^2$ by showing the first eight (undamped) density functions for both sets in different colors.
\begin{figure}[!ht]
  \centering
  \parbox{45mm}{
  \vspace{-0.5em}
  \includegraphics[width=40mm]{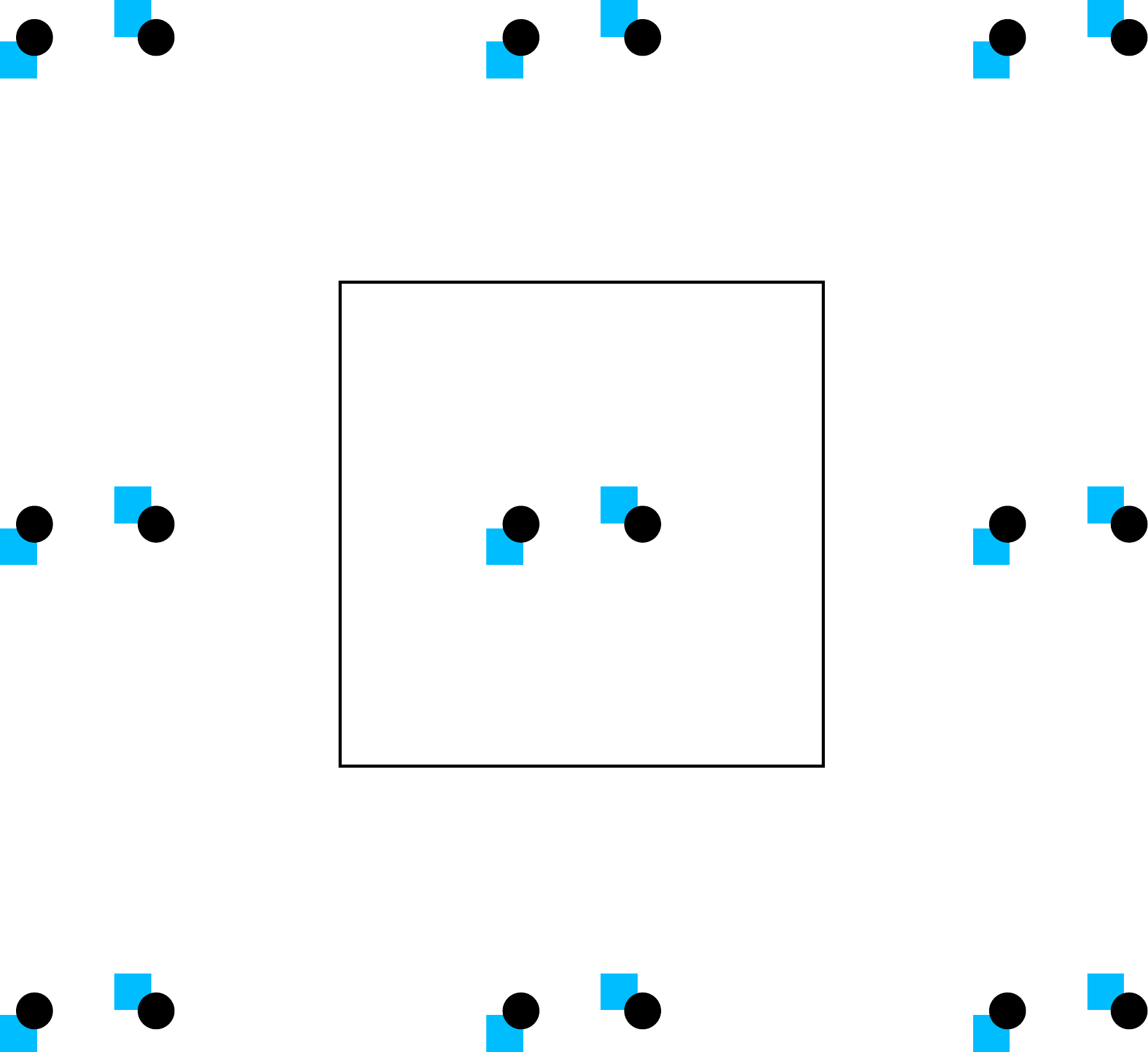}}
  \hspace{2mm}
  \parbox{80mm}{
  \includegraphics[width=80mm]{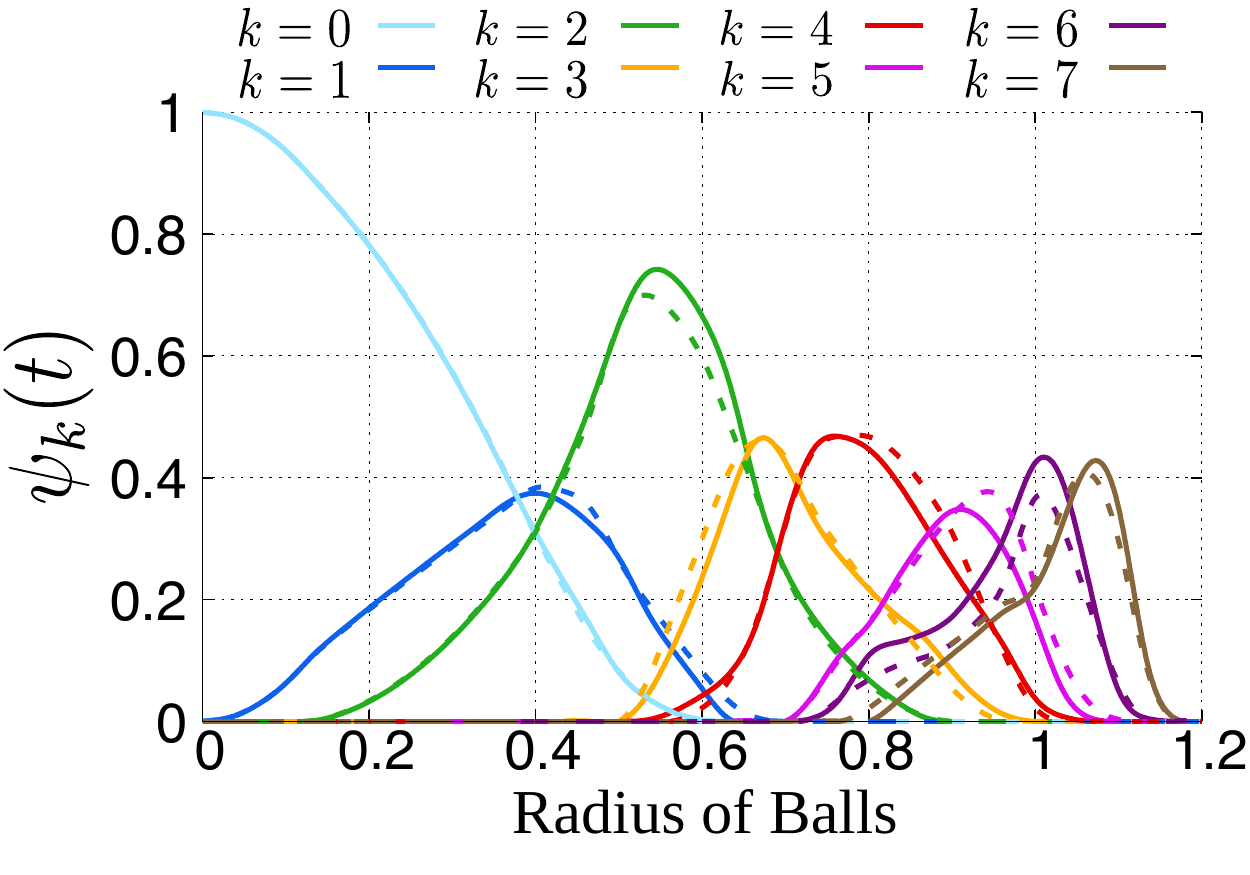}}
  \caption{
    \emph{Left}: a periodic set with two \emph{black} points in its square unit cell, and the perturbed periodic set with two \emph{blue} points in the same unit cell. 
    \emph{Right}: the graphs of the density functions are \emph{solid} for the original set and \emph{dashed} for the perturbed set.
    As predicted by Theorem~\ref{thm:fingerprint_stability}, the small perturbation of the periodic set causes a small change in the fingerprint.}
\label{fig:cloud2D_packing2}
\end{figure}

\section{Completeness}  
\label{sec:Completeness}

The fingerprint map is \emph{complete} if it is injective up to isometries; that is: non-isometric periodic sets are mapped to different fingerprints.
We prove completeness generically, i.e.\ on a dense open subset; compare to \cite{BoKe04,Sen08,LSS03}.
The density fingerprint also distinguishes non-generic sets for which other means fail, as will be illustrated by an example in Section \ref{sec:hope_removing_genericity}.
The completeness of the fingerprint for all periodic sets, however, remains an open question.
Indeed, at the time of writing this paper, the authors
are not aware of a $3$-dimensional counterexample to completeness, but there is a $1$-dimensional counterexample due to Morteza Saghafian:
letting $U=\{0,4,9\}$ and $V=\{ 0,1,3\}$, it can be checked 
that the finite sets $U+V$ and $U-V$, and the periodic sets $15 \Zspace + (U+V)$ and $15 \Zspace + (U-V)$ cannot be distinguished by the $1$-dimensional density fingerprint map.

\subsection{Generic Completeness}
\label{sec:generic_completeness_alternative}

We prove the completeness of the density fingerprint map for generic periodic sets in $\Rspace^3$.
The notion of genericity is defined by conditions that are satisfied by an open and dense subset of the space of periodic point sets. 
We formulate such conditions in terms of the \emph{circumradius} of edges, triangles, and tetrahedra, 
which is the radius of the smallest sphere that passes through the vertices of the simplex. 
To avoid infinitely many constraints, we introduce an upper bound on the circumradii to consider. Denoting by $L=L(A,\vartheta)$ the list of all edges (pairs), triangles (triples) and tetrahedra (quadruples) spanned by points of a periodic set $A$ whose circumradius is at most $\vartheta$, we 
call $A$ \emph{generic for a constant threshold, $\vartheta$}, if---apart from necessary violations due to periodicity---it satisfies the following three conditions:
\medskip \begin{enumerate}[ii]
  \item[I.] the circumradii of different simplices in $L$ are different;
  \item[II.] the circumradii of different edges in $L$ are not related to each other by a factor of $2$.
  \item[III.] for every $t \leq \vartheta$, there is at most one set of six circumradii of simplices in $L$ such that the edges with twice their lengths assemble to a tetrahedron whose circumradius is $t$.
\end{enumerate} \medskip
We call an edge a \emph{lattice edge} if its length is 
the distance between two lattice points.
Lattice edges violate Condition II and can thus be identified as such.
A \emph{lattice triangle} has three lattice edges, and a \emph{lattice tetrahedron} has six lattice edges. 
The important difference between lattice and non-lattice simplices is that only the latter are unique up to lattice translations.

\medskip
Since Conditions I, II, III can be phrased via finitely many algebraic equations in the vectors $x\in M, v_1,v_2,v_3$, the set of generic periodic sets with threshold $\vartheta$ is is open and dense in the space of all periodic sets with at most $m$ motif points.
We write $\Radius{A}$ for the largest finite circumradius of $p \leq 4$ points in $A$ with pairwise distance at most four times the diameter of the unit cell.
Since the diameter is the distance between two lattice points, this implies that $\Radius{A}$ is at least double the diameter.
\begin{theorem}[Generic Completeness]
  \label{thm:generic_completeness_alternative}
  Let $A,Q \subseteq \Rspace^3$ be non-isometric periodic sets that are 
  generic for the threshold $\vartheta = \max\{\Radius{A},\Radius{Q}\}$. Then $\Fingerprint{}(A) \neq \Fingerprint{}(Q)$. 
\end{theorem}
\begin{proof}
  Let $[A]$ denote the isometry class of $A$.
  We prove the unique reconstruction of $[A]$ from $\Fingerprint{}(A)$ in two steps: 
  \begin{align*}
    \Fingerprint{}(A) \hspace{0.3em} \xrightarrow{\textrm{\sc Step } 1} \hspace{0.3em} \textrm{tetrahedra in $L(A,\vartheta)$, up to isometries} \hspace{0.3em} \xrightarrow{\textrm{\sc Step } 2} \hspace{0.3em} [A]
  \end{align*}
  \noindent
  {\sc Step 1:} Each density function is a weighted sum of the volumes of intersections of $2$, $3$, or $4$ balls around points of $A$; see \cite[Equation (5)]{EdIg18}.
  The volume formulas of such intersections are given in \cite{EdFu94}.
  It is cumbersome but not difficult to prove that they are piecewise analytic, and that the circumradii of edges, triangles and tetrahedra spanned by points of $A$ are the positions where the functions are not analytic. 
Therefore, the set of all positions up to $\vartheta$ where at least one density function is not analytic yields the set of circumradii of simplices of $L$.
We avoid the technicalities of using the differences between the left- and right-derivatives to distinguish which of these are caused by $2$, $3$, or $4$ balls meeting, with the following trick. We treat all circumradii as if they were circumradii of edges, multiply them by two (to get the edge length), and try to assemble six of these edge lengths to form a tetrahedron.
Whenever this gives a circumradius of a simplex of $L$, we have found a tetrahedron of $A$ by Condition III. This way we can uniquely construct all tetrahedra of $L$ up to isometries.

\medskip \noindent
{\sc Step 2:} To start the process, we choose a non-lattice tetrahedron from the list.
If there is no such tetrahedron, then $A$ is a lattice and can be reconstructed from the lexicographically shortest lattice tetrahedron from the list---i.e.\ the tetrahedron consisting of the shortest lattice edge, the second-shortest lattice edge (linearly independent from the first), and so on---defining a (Minkowski-)reduced \cite{NgSt04} 
and therefore primitive unit cell of $A$.
On the other hand, if there exist non-lattice tetrahedra, we choose the lexicographically shortest one, $abcd$, with non-lattice edge $ab$. 

Placing $abcd$ in space---as we are only interested in the isometry class of $A$, we can place it arbitrarily---we identify all tetrahedra $abce$ from the list that have $abc$ as a face and try to glue them onto $abcd$. There are two possibilities (related by a reflection) of how to glue $abce$; we denote the two different tip positions by $e_1$ and $e_2$. We prove that at most one of the two options
gives a positive result when checking if the tetrahedron $abde_i$ is in the list of tetrahedra from Step $1$: The triangles $abd$ and $abe$ 
are non-lattice, and therefore unique in $A$ up to lattice translations by Condition I. 
Thus, when glued along $ab$, they span a uniquely defined tetrahedron $abde$ with a certain edge length $de$ that is the distance between $d$ and $e_i$ for at most one of its two possible positions.

\medskip
This gluing procedure yields (among others) all points of distance at most four times the diameter of the unit cell 
to $a,b,c,d$ (by definition of $\Radius{A}$), except the ones that lie on a plane spanned by the triangles 
$abc$ or $abd$. This neighborhood is large enough such that it contains every motif point at least once and such that it contains a lattice basis, which can be identified by computing the pairwise differences between the reconstructed points and checking whether they satisfy Condition II. Repeating the reconstructed points with respect to the lattice yields the isometry class of $A$. As the construction was unique given the genericity conditions, we get $\Fingerprint{}(A)\neq \Fingerprint{}(Q)$.
\end{proof}

\subsection{Distinguishing Non-Generic Periodic Sets}
\label{sec:hope_removing_genericity}

There are indications that the density fingerprint map distinguishes all periodic sets and not just the generic ones.
We now give the reason for our optimism. Example~\ref{eg:autocorrelation} describes two periodic sets that violate the above genericity conditions and can nevertheless be distinguished by the density fingerprint map.
On the other hand, the two sets can neither be distinguished by their density nor by their X-ray diffraction patterns; two means commonly used in crystallography to determine the structure of a crystal.
X-ray diffraction patterns give all pairwise distance vectors of the periodic set, but they do not determine the isometry class of a periodic set \cite{PaSh30}: there exist \emph{homometric structures}, which are non-isometric periodic sets with the same $2$-point autocorrelation functions; that is: identical multisets of pairs, up to 
translation.
There even exist periodic sets with the same $2$- and $3$-point autocorrelation functions, as we now explain.

\begin{figure}[ht]
  \includegraphics[height=33mm]{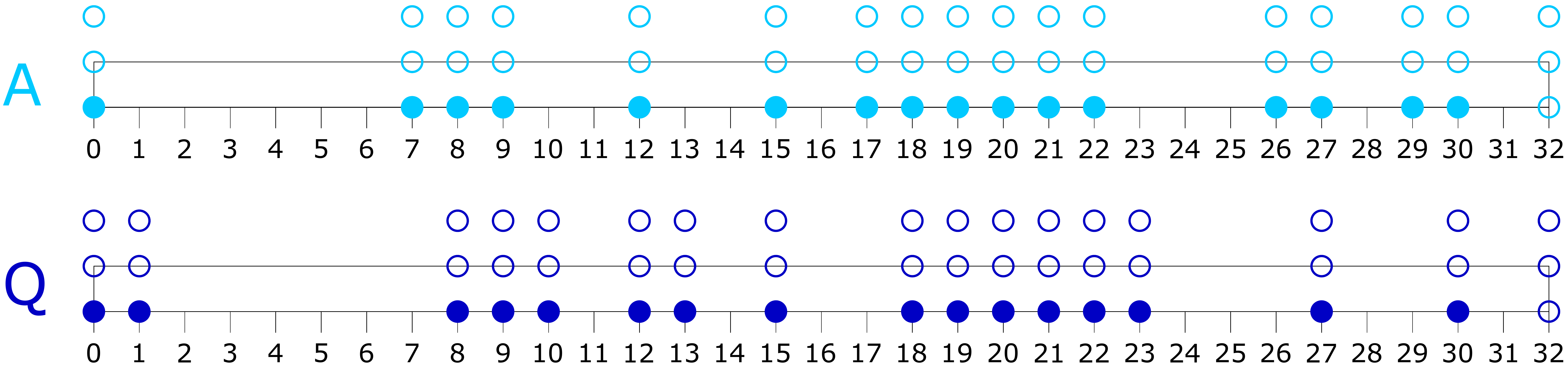}
  \caption{Periodic sets $A$ and $Q$ from Example~\ref{eg:deng_moody}, pictured with rectangular unit cells in two dimensions, for simplicity.
  Filled dots belong to the motifs while unfilled dots show the periodicity.}
  \label{fig:deng_moody}
\end{figure}

\begin{example}
\label{eg:autocorrelation}
  Let $A^{(1)}$
  and $Q^{(1)}$ be sets with periodicity $32$ in $\mathbb{R}$, each with $16$ points in the corresponding motif:
  \begin{align}
    0, 7, 8, 9, 12, 15, 17, 18, 19, 20, 21, 22, 26, 27, 29, 30; \\
    0, 1, 8, 9, 10, 12, 13, 15, 18, 19, 20, 21, 22, 23, 27, 30;
  \end{align}
  see Figure \ref{fig:deng_moody}.
  The authors of \cite[Section 5.3]{GrMo95} show that $A^{(1)}$ and $Q^{(1)}$ have the same $2$- and $3$-point autocorrelation functions.
  Taking the Cartesian product with $\Zspace^2$ preserves the equality between the autocorrelation functions, which yields periodic sets, $A,Q \subseteq \Rspace^3$, with matching $2$- and $3$-point autocorrelation functions.
  Nevertheless, our density fingerprint map distinguishes them, as shown in Table~\ref{tab:autocorrelation}: the $L_\infty$-distances between the first four corresponding density functions vanish but the next five $L_\infty$-distances
  are strictly positive.
  \label{eg:deng_moody}
\end{example}
    
\begin{table}[ht]
  \begin{center}
  \begin{tabular}{c||ccccccccc}
    $k$ & 0 & 1 & 2 & 3 & 4 & 5 & 6 & 7 & 8 \\ \hline \hline
    $\infdist{\probB{k}{A}}{\probB{k}{Q}}$ & 0.000 & 0.000 & 0.000 & 0.000 & 0.005 & 0.007 & 0.013 & 0.022 & 0.007
  \end{tabular}
  \end{center}
  \caption{$L_\infty$-distances between the corresponding density functions of the sets $A$ and $Q$ in Example~\ref{eg:autocorrelation}.}
  \label{tab:autocorrelation}
\end{table}

\section{Computation} 
\label{sec:Computation}

The algorithm for the density fingerprint map is based on two related geometric concepts: the $k$-th Dirichlet--Voronoi domain and the $k$-Brillouin zone of a point.
After introducing both, we explain how they are used, and how much time it takes to construct them.

\subsection{Dirichlet--Voronoi Domains and Brillouin Zones}
\label{sec:Dirichlet-Voronoi_Domains_and_Brillouin_Zones}

Let $A \subseteq \Rspace^3$ be a locally finite set of points.
For every positive integer $k$, the \emph{$k$-th Dirichlet--Voronoi domain} of a point $a \in A$ is the set of points in $\Rspace^3$ for which $a$ is among the $k$ closest points in $A$, and the \emph{$k$-th Brillouin zone} is the difference between the $k$-th and the $(k-1)$-st Dirichlet--Voronoi domains:
\begin{align}
  \domain{k}{a}{A}  &=  \{ x \in \Rspace^3 \mid \Edist{x}{b} < \Edist{x}{a} \mbox{\rm ~for at most~} k-1 \mbox{\rm ~points~} b \in A \} , \\
  \zone{k}{a}{A}  &=  \domain{k}{a}{A} \setminus \domain{k-1}{a}{A} ;
\end{align}
see Figure \ref{fig:brillouin}.
Here we set $\domain{0}{a}{A} = \emptyset$ so that the first Brillouin zone is well defined.
Note that $\zone{k}{a}{A}$ is the set of points $x \in \Rspace^3$ for which there are exactly $k-1$ points $b \in A$ that are closer to $x$ than $a$ is.
Observe also that $\domain{k}{a}{A}$ is closed and star-convex, and if $A$ is Delone, then it is also compact.
If $A$ is a lattice, $A = \Lambda$, then all $k$-th Dirichlet--Voronoi domains are translates of each other, and similarly for the Brillouin zones: $\domain{k}{a}{\Lambda} = \domain{k}{0}{\Lambda} + a$ and $\zone{k}{a}{\Lambda} = \zone{k}{0}{\Lambda} + a$.
Except for a measure zero subset of $\Rspace^3$, every point $x$ has a unique $k$-closest point in $\Lambda$.
This implies that the $k$-th Brillouin zones \emph{tile} $\Rspace^3$, by which we mean that their closures cover $\Rspace^3$ while their interiors are pairwise disjoint.
These properties generalize to a periodic set, $A = M + \Lambda$:  the $k$-th Brillouin zones of the points in $a + \Lambda$ are translates of each other, and the $k$-th Brillouin zones of all $a \in A$ tile $\Rspace^3$.
\begin{figure}
  \centering
  \includegraphics{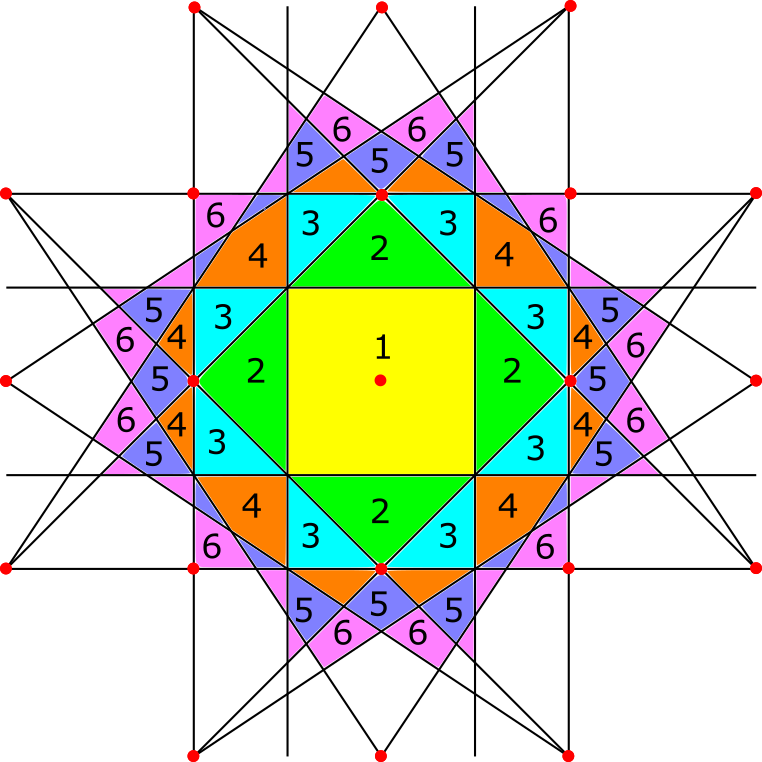}
  \caption{
    The $6$-th Dirichlet--Voronoi domain of the point in the center decomposed into the first six Brillouin zones, which are indicated by colors and labels.
  }
  \label{fig:brillouin}
\end{figure}

\subsection{Decomposed Multiple Cover}
\label{sec:Decomposed_k-fold_Cover}

Assume from here on that $A = M + \Lambda$ is a periodic set.
To compute $\probA{k}{A}$, we may use any fundamental domain of the lattice.
Particularly convenient is the union of the $k$-th Brillouin zones of the points in $M$ as it lends itself to finding the subset covered by at least $k$ of the balls.
\begin{theorem}[Density for Periodic Set]
  \label{thm:density_for_periodic_set}
  Let $A = M + \Lambda$ be periodic with lattice $\Lambda \subseteq \Rspace^3$
  and motif $M \subseteq U$ in the unit cell of $\Lambda$, and let $k \geq 1$ be an integer.
  Then the probability that a random point $x \in U$ belongs to
  $k$ or more balls of radius $t \geq 0$ centered at the points of $A$ is
  \begin{align}
    \probA{k}{A}(t)  &=  \tfrac{1}{\volume{U}} \sum_{a \in M} \volume{\zone{k}{a}{A} \cap \ball{t}{a}}.
    \label{eqn:thm8}
  \end{align}
\end{theorem}
\begin{proof}
  Let $M_k$ be the union of the $k$-th Brillouin zones of the points $a \in M$ and note that $M_k + \Lambda$ tiles $\Rspace^3$.
  It follows that $\volume{M_k} = \volume{U}$.
  Let $x \in M_k$ be in the interior of $\zone{k}{a}{A}$.
  By construction, $a$ is the unique $k$-closest point to $x$, so $x$ lies in $k$ or more balls if and only if $x \in \ball{t}{a}$.
  Summing over all points $a \in M$ gives \eqref{eqn:thm8}.
\end{proof}
Clearly, $\probA{0}{A} (t) = 1$ for all radii $t$.
Given $k \geq 0$ and $t \geq 0$, we use \eqref{eqn:thm8} to compute $\probA{k}{A} (t)$ and $\probA{k+1}{A} (t)$, and we get $\probB{k}{A} (t) = \probA{k}{A} (t) - \probA{k+1}{A} (t)$.
To implement \eqref{eqn:thm8}, we need to compute the volume of the intersection of a ball with a convex polyhedron.
We could, for example, decompose the polyhedron into tetrahedra and use explicit expressions for the volume of intersections between balls and simplices; see for example \cite{ABI88}.
A C++ implementation evaluating the density functions for a given periodic set using this strategy can be found at \cite{Sm20}.
Alternatively, we could use inclusion-exclusion, which allows for further consolidation of the formula, writing $\probA{k}{A} (t)$ as an alternating sum of common intersections of up to four balls each.
This does not lead to any asymptotic improvements of the running time, so we omit further details and refer to \cite{EdIg18} instead.

\subsection{Algorithm and Running Time}
\label{sec:Algorithm_and_Running_Time}

To evaluate the density functions $\probB{0}{A}, \probB{1}{A}, \ldots, \probB{k}{A}$ at a value $t$, we compute a plane arrangement for each point $a \in M$ that consists of enough planes so that the first $k+1$ Brillouin zones of $a$ occur.
Specifically, for a large enough radius, $s$, we consider for each $b \neq a$ in $A \cap \ball{s}{a}$ the \emph{bisector} of $b$ and $a$, which is the plane defined by $\Edist{x}{a} = \Edist{x}{b}$.
These bisectors decompose $\Rspace^3$ into convex cells.
We refer to this decomposition as the \emph{arrangement} of the planes.
The $3$-dimensional cells that are separated from $a$ by exactly $j-1$ planes form the \emph{$j$-th belt} of the arrangement.

\medskip
We now address the question how small we can choose $s$ such that the first $k+1$ belts are the first $k+1$ Brillouin zones of $a$.
To begin, we recall that $t \geq 2 R \sqrt[3]{k+1}$ implies that $\probB{i}{} (t) = 0$ for $0 \leq i \leq k$;
see the proof of Theorem \ref{thm:fingerprint_stability}.
To express this insight geometrically, let $R_{k+1} (a)$ be the maximum distance of a point in the $(k+1)$-st Brillouin zone of $a$ from $a$.
That the density functions $\probB{0}{A}$ to $\probB{k}{A}$ are zero for $t \geq 2R \sqrt[3]{k+1}$ implies $\probA{k+1}{A} (t) = 1$, for these values of $t$, and therefore $R_{k+1}(a) \leq 2 R \sqrt[3]{k+1}$.
To capture all the relevant planes, it thus suffices to consider all points $b \in A \setminus \{a\}$ at distance at most $s = 2 R_{k+1}(a)$ from $a$.
Using a straightforward volume argument, we see that $\ball{2 R_{k+1} (a)}{a}$ contains at most $(4 R \sqrt[3]{k+1} + r)^3 / r^3 = \mathcal{O}(k)$ points, in which we treat $r$ and $R$ as constants.

\medskip
Constructing the arrangement of $\mathcal{O}(k)$ planes incrementally, as described in \cite[chapter 7]{Ede87}, takes time $\mathcal{O}(k^3)$.
Doing this for each point in the motif takes time $\mathcal{O}(\card{M} \cdot k^3)$,
and within the same time bound we can evaluate the first $k+1$ density functions. 

\section{An Application to Crystal Structure Prediction} 
\label{sec:Appication}

Crystal Structure Prediction (CSP) aims to predict whether a selected molecule can exist as a functional material, i.e.\ 
a crystal with useful functions or properties. 
In other words, CSP seeks to answer the question of whether copies of a molecule can be arranged in such a way that the resulting crystal is stable (will not deform and lose its properties over time) as well as useful.
Crucially, CSP tries to answer this question without setting foot in a laboratory, with the hope of dramatically reducing the need to perform the time-consuming process of physically synthesizing crystals.

\begin{figure}[ht]
\vspace{0.1in}
\includegraphics[height=38mm]{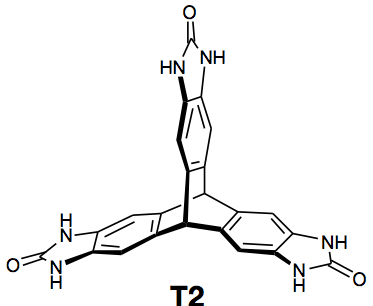}
\includegraphics[height=38mm]{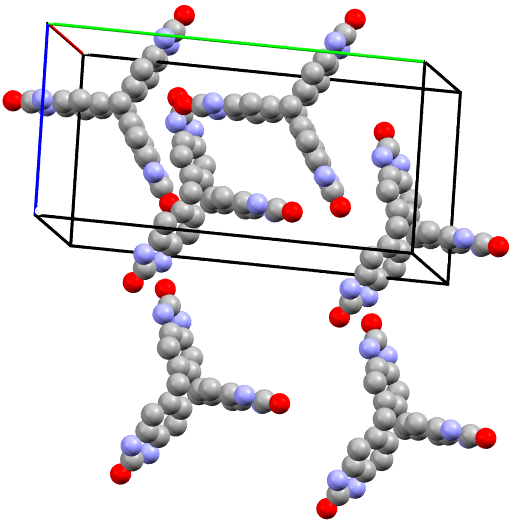} \hspace{0.1in}
\includegraphics[height=38mm]{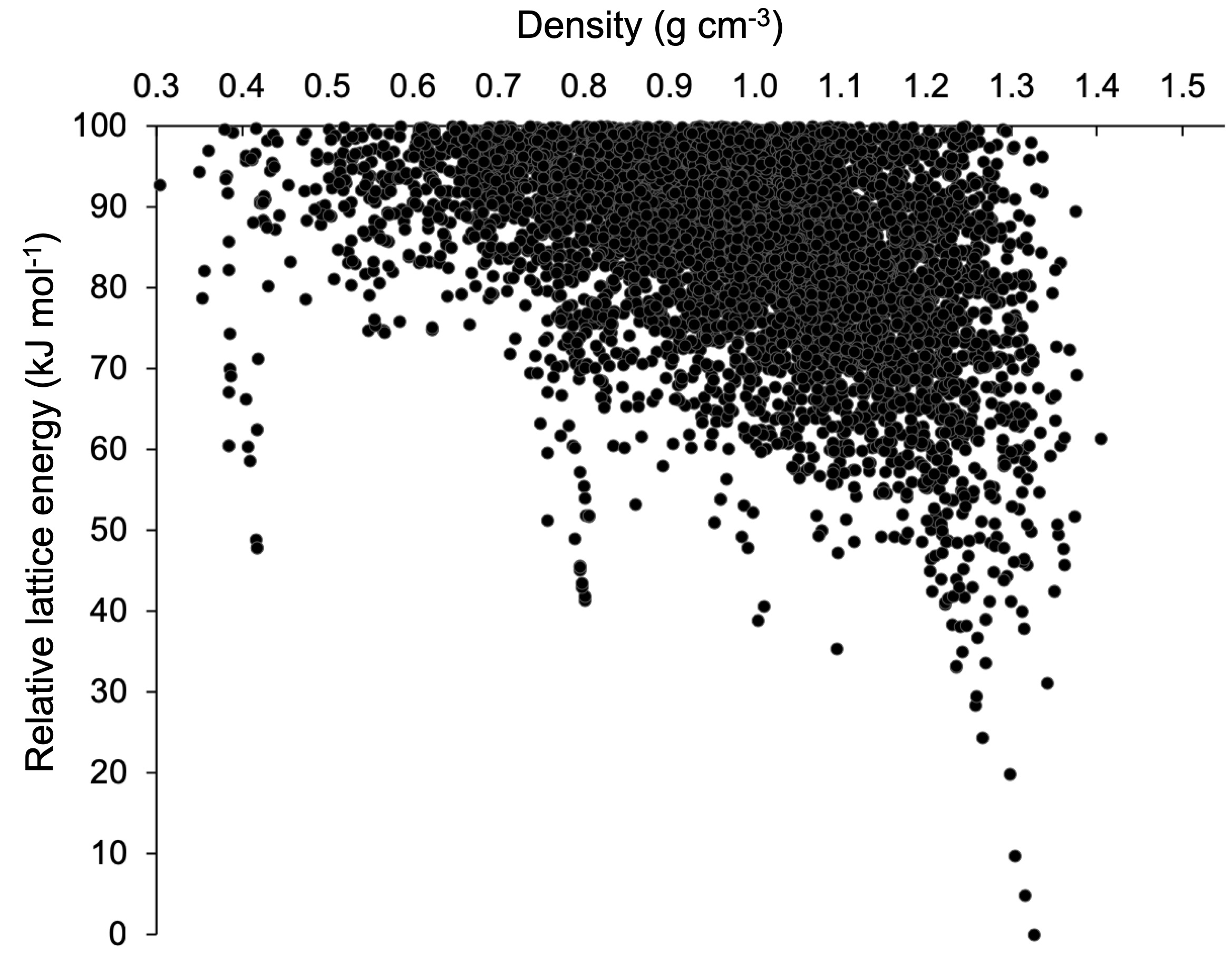}
\caption{\emph{Left}: a T2 molecule. 
  \emph{Middle}: the T2-$\delta$ crystal with highlighted unit cell. 
  \emph{Right}: the output of CSP for the T2 molecule.
  It is a plot of 5679 simulated T2 crystal structures \cite[Fig.~2d]{Pul17}, each represented by two coordinates: the physical density (atomic mass within a unit cell divided by the unit cell volume) and energy (determining the crystal's thermodynamic stability). 
  Structures at the bottom of the `downward spikes' are likely to be stable.}
\label{fig:T2molecule}
\end{figure}

\medskip
Our collaborators at Liverpool's Materials Innovation Factory \cite{Pul17} used CSP to predict that the T2 molecule (Figure~\ref{fig:T2molecule}) can be crystallized into a new structure that has half the physical density of the only previously known structure for T2, a desirable property for applications such as gas storage. As part of this process, they also identified four other structures of interest.
Following the CSP predictions, they synthesized $5$ families of T2-crystals in the laboratory by varying parameters like temperature and pressure, calling them T2-$\alpha$, T2-$\beta$, $\dots$, T2-$\epsilon$. One of them, T2-$\gamma$, indeed had the desired property of having only half the physical density of the previously known structure T2-$\alpha$.
They scanned the synthesized crystals using X-ray powder diffraction yielding Crystallographic Information Files, 
each containing the unit cell and the motif points representing the atoms. 
These files were then compared with the results of the simulations,  
either by using their physical density alongside the \textsc{Compack} algorithm---which compares only a finite portion of the structure---or by looking at visualizations 
of the crystal structures.
This comparison showed that the synthesized crystals matched the prediction well.
Our collaborators 
deposited these structures into the globally used Cambridge Structural Database.

\medskip
At a later time, we used our newly developed fingerprints to verify our collaborators' matchings between the synthesized crystals T2-$\alpha$ to T2-$\epsilon$ and the simulated crystals entry 99, 28, 62, 09, 01.
We did so by computing, for each of the five matches, the distance between the density functions of the synthesized and the simulated crystal.
As one is the prediction of the other, we expected to see small distances.
And for four of the five structures this was true: T2-$\gamma$, for example, always has an $L_{\infty}$-distance of less than $0.04$ over the first eight pairs of corresponding density functions; see Table~\ref{tab:T2distances}.
However, when we came to check the distances between density functions of T2-$\delta$ with its predicted structure, we were surprised to see large distances (the final row of Table~\ref{tab:T2distances}).
It turned out that a mix-up of files had happened, and what was uploaded to the Cambridge Structural Database as T2-$\delta$ was in fact T2-$\beta'$ (a crystal from the T2-$\beta$ family). 
The density fingerprint revealed this error, which was verified by chemists upon a visual inspection, and it is 
because of this that T2-$\delta$ was subsequently correctly deposited. 

\begin{table}[ht]
  \begin{center}
  \begin{tabular}{c||cccccccc}
    $\infdist{\probB{k}{A}}{\probB{k}{Q}}$ & $k=0$ & 1 & 2 & 3 & 4 & 5 & 6 & 7 \\
    \hline \hline
    T2-$\alpha$ vs entry 99   & 0.0042 & 0.0092 & 0.0125 & 0.0056 & 0.0099 & 0.0088 & 0.0127  & 0.0099 \\
    T2-$\beta$ vs entry 28    & 0.0157 & 0.0156 & 0.0159 & 0.0224 & 0.0334 & 0.0396 & 0.0357  & 0.0454 \\
    T2-$\gamma$ vs entry 62   & 0.0020 & 0.0080 & 0.0128 & 0.0155 & 0.0153 & 0.0250 & 0.0296  & 0.0391 \\
    T2-$\delta$ vs entry 09   & 0.0610 & 0.0884 & 0.1267 & 0.0676 & 0.0915 & 0.0801 & 0.0733  & 0.0388 \\
    T2-$\epsilon$ vs entry 01 & 0.0132 & 0.0152 & 0.0207 & 0.0571 & 0.0514 & 0.0431 & 0.0468  & 0.0550 \\ \hline
    T2-$\beta'$ vs entry 09   & 0.2981 & 0.2631 & 0.3718 & 0.3747 & 0.2563 & 0.2360 & 0.3161  & 0.3232
 \end{tabular}
 \end{center}
  \caption{\textit{First five rows:} the $L_\infty$-distances between the first eight pairs of corresponding density functions of physically synthesized T2 crystals (T2-$\alpha$, T2-$\beta$, etc.) and the simulated structures that had predicted them from the CSP output dataset (entry XX).
  \textit{Last row:} the suspiciously larger numbers revealed the mix-up of the files T2-$\delta$ and T2-$\beta'$ and thus led to depositing the initially omitted Crystallographic Information File of the T2-$\delta$ crystal into the Cambridge Structural Database. 
  }
  \label{tab:T2distances}
\end{table}

Plots of the density functions of correctly matched synthesized and simulated structures can be seen in Figure~\ref{fig:T2densities}. 
As another application, we expect that the fingerprint will 
be used 
to simplify the large output data sets produced by CSP by comparing simulated structures with each other, thus speeding up what is currently a slow process.

\newcommand{\cheight}{34mm}

\begin{figure}[!ht]
  \centering
\includegraphics[height=\cheight]{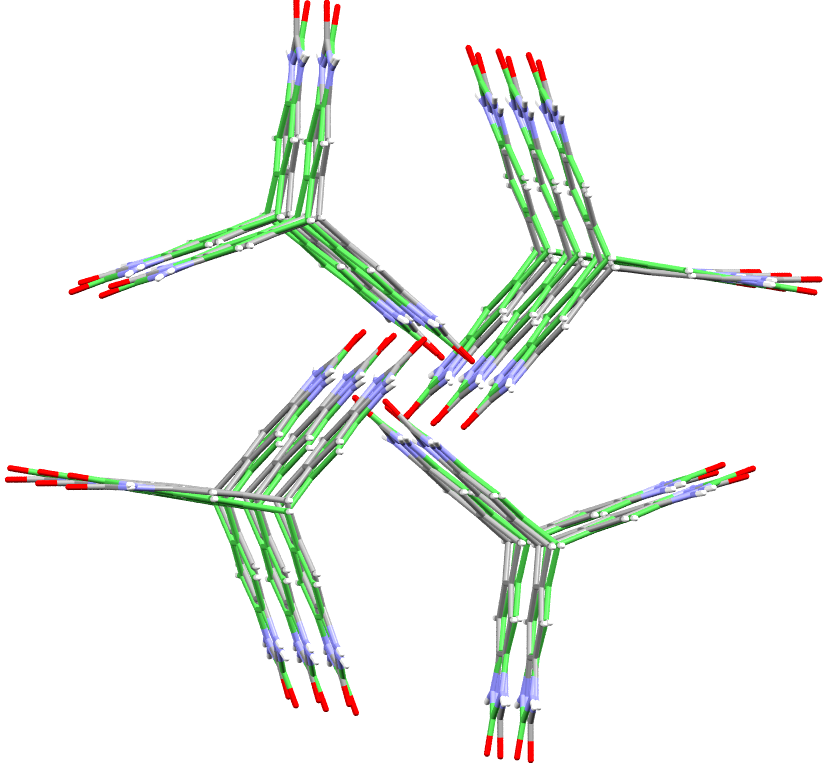} \hfill
  \includegraphics[height=\cheight]{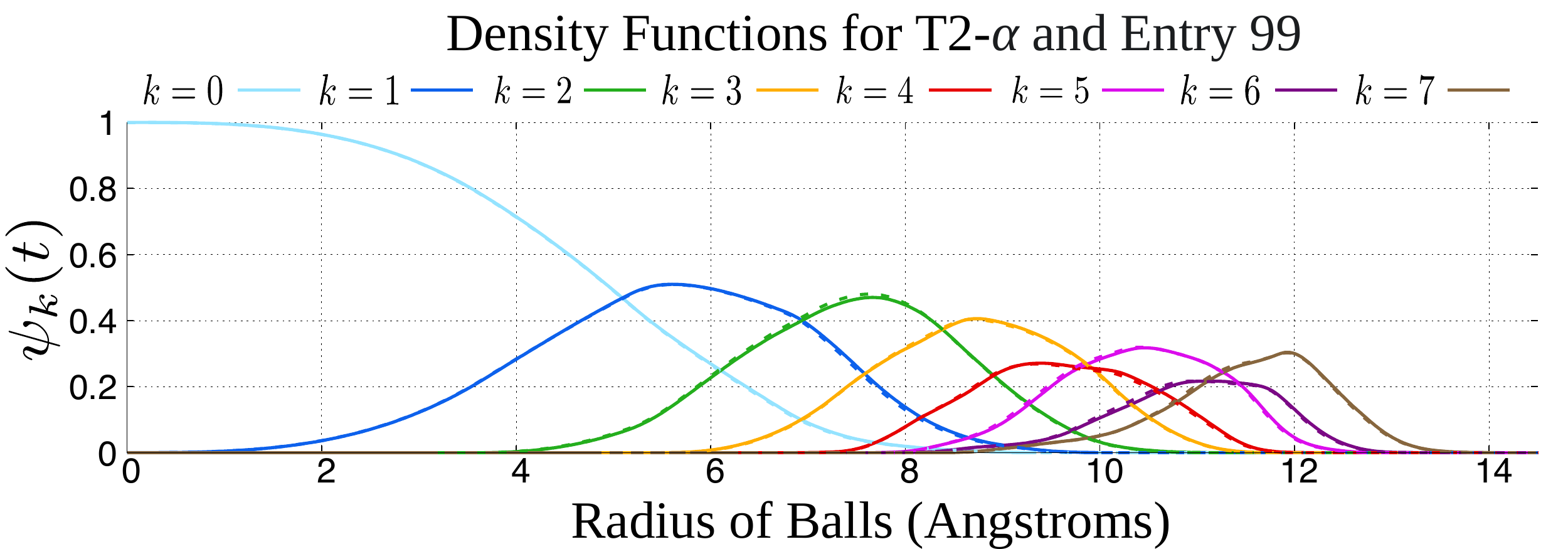} 
\smallskip
  
\includegraphics[height=\cheight]{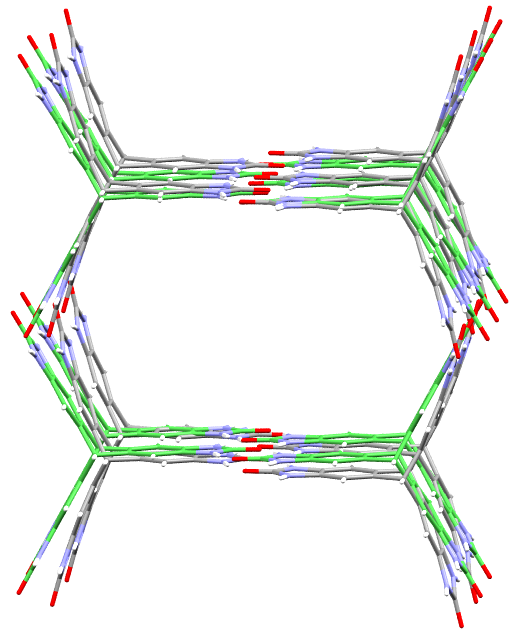} \hfill
  \includegraphics[height=\cheight]{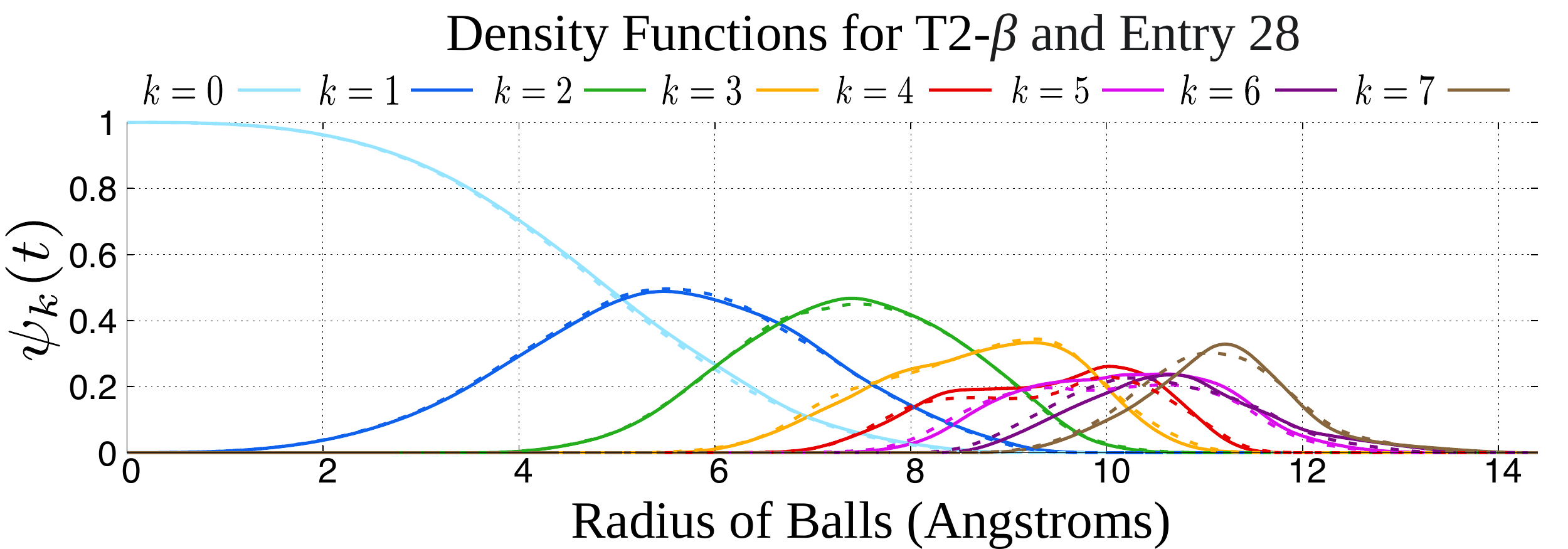} 
\medskip

\includegraphics[height=\cheight]{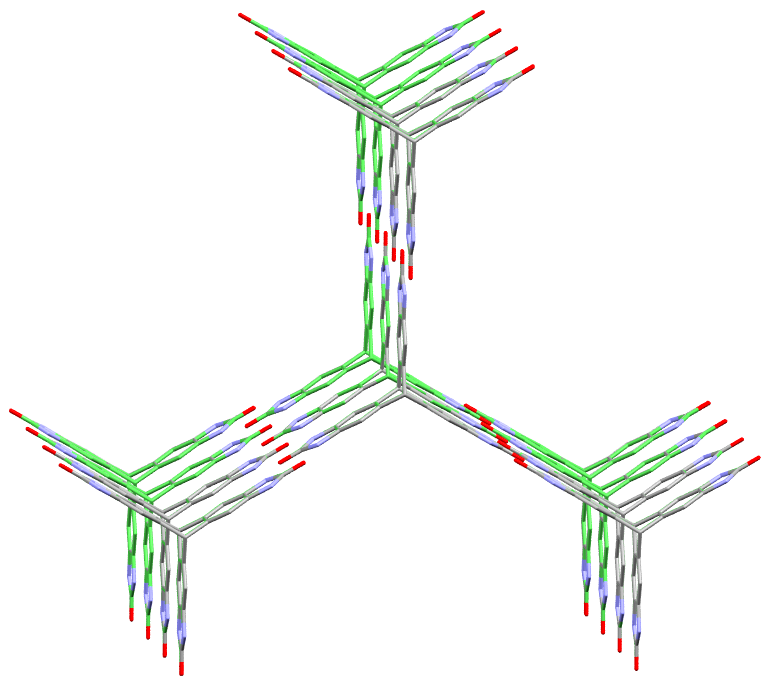} \hfill
  \includegraphics[height=\cheight]{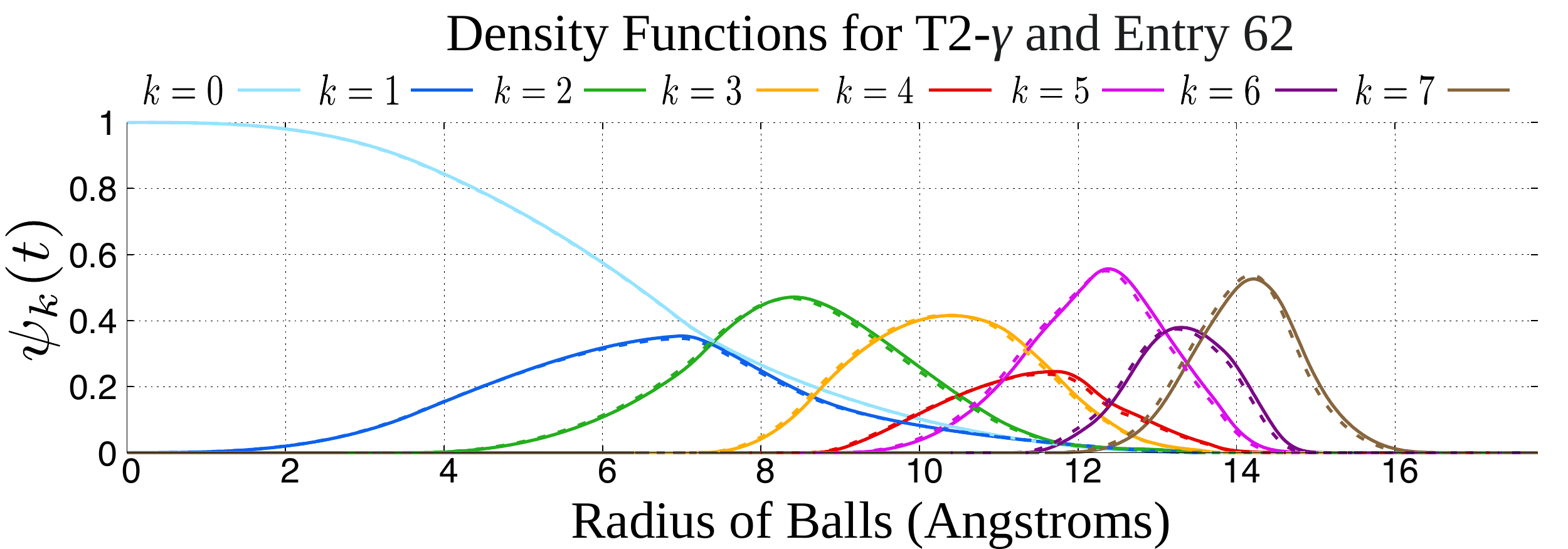} 
\medskip
  
\includegraphics[height=\cheight]{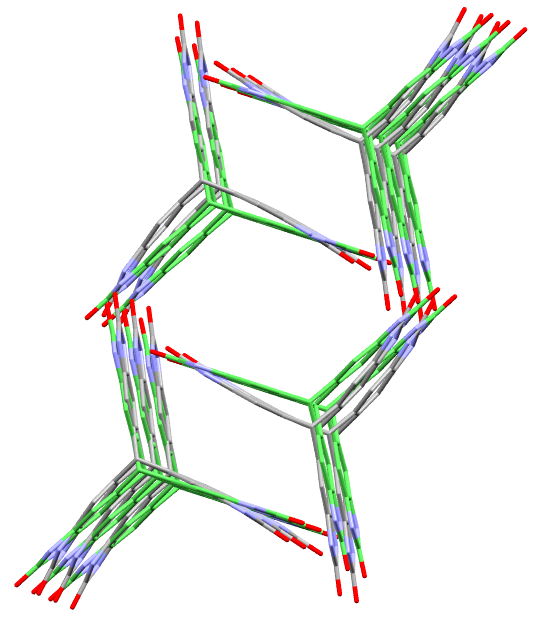} \hfill
  \includegraphics[height=\cheight]{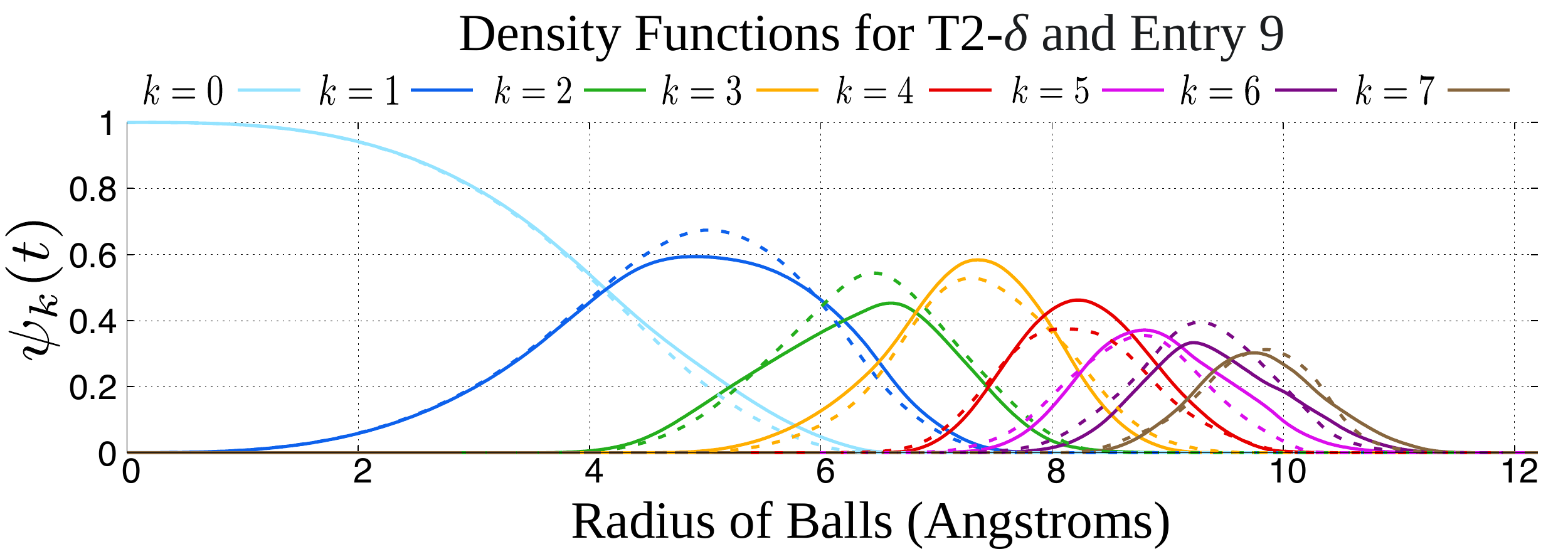}
\medskip
  
\includegraphics[height=\cheight]{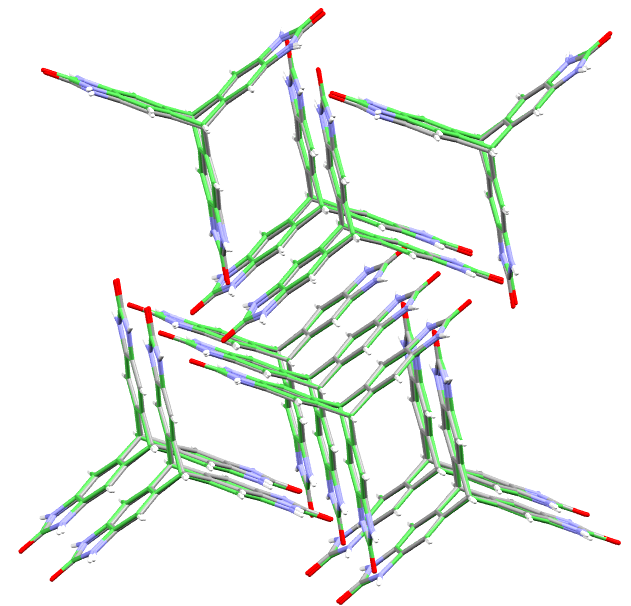} \hfill
  \includegraphics[height=\cheight]{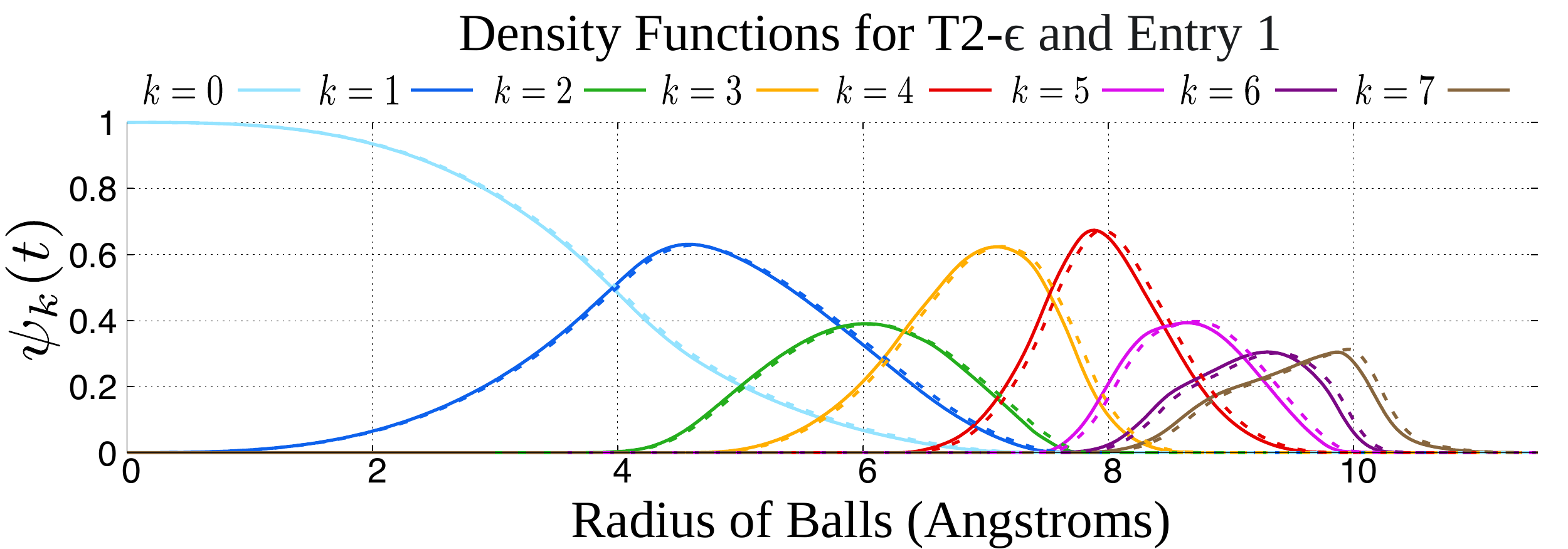}
\caption{
  \emph{Left}: experimental T2 crystals (curved gray molecules) and their simulated versions (straight green molecules) overlaid. 
  \emph{Right}: the density functions of the periodic sets of molecular centres of the experimental T2 crystals (solid curves) vs.\ simulated crystals (dashed curves).}
\label{fig:T2densities}
\end{figure}

\section{Discussion} 
\label{sec:Discussion}

The main contribution of this paper is a fingerprint map from periodic sets in $\Rspace^3$ (which model crystals) to series of density functions.
This map is obviously invariant under isometries, and we prove it is continuous
and generically complete.
We leave the completeness without genericity assumption as an open question.
In this context, it is worth noticing that our proof of generic completeness makes only limited use of the order, $k$, at which the circumradius of an edge, triangle, or tetrahedron is detected.
Recall that the order is the number of points in the respective circumsphere.
Is this additional information sufficient to prove general completeness?

\medskip
A drawback of the bottleneck distance between periodic sets used in this paper is its sensitivity to changes of the unit cell; see Lemma~\ref{lem:common_lattice}.
An alternative dissimilarity that may be more relevant in practice considers affine transformations, $\tau$, that minimize the bottleneck distance:
\begin{align}
  d_{\rm AT} (A, Q)  &=
    \inf_{\tau} \max \{
    \min \{ \dBottleneck{A}{\tau (Q)} ,
            \dBottleneck{\tau (A)}{Q} \} ,
    |\log s_1|, |\log s_3| \},
\end{align}
in which $s_1 \geq s_2 \geq s_3$ are the three singular values of the matrix of $\tau$.
Is the density fingerprint map defined in Section \ref{sec:Fingerprint} continuous with respect to this dissimilarity?

\medskip
We close this paper with three extensions of the results presented in this paper.
Different types of atoms are often modeled as balls with different radii.
A possible geometric formalism is that of weighted points and the power distance \cite{Aur87}.
Our geometric results generalize to this setting, but some need a careful adaptation.
Our continuity result for periodic sets (Theorem \ref{thm:fingerprint_stability}) also generalizes to non-periodic Delone sets that allow for a reasonable definition of density functions.
Considering that quasiperiodic crystals can be modeled as such, it might be worthwhile to find out how far such an extension can be pushed.
Finally, we mention that our results generalize to
arbitrary finite dimension.

\bibliography{biblio}

\clearpage

\end{document}